\documentclass[USenglish,oneside,twocolumn]{article}

\usepackage[T1]{fontenc}
\usepackage[utf8]{inputenc}
\usepackage[big]{dgruyter_NEW}

\usepackage{lmodern}
\usepackage{microtype}
\usepackage{amsthm}
\usepackage{amssymb}
\usepackage{stmaryrd}
\usepackage{graphicx}
\usepackage{hyphenat}

\usepackage{tikz}
\usetikzlibrary{%
  calc,%
  fit,%
  positioning,%
  shapes.geometric,%
}
\pgfdeclarelayer{background}
\pgfsetlayers{background,main}

\theoremstyle{plain}
\newtheorem{thm}{\protect\theoremname}
\theoremstyle{plain}

\theoremstyle{definition}
\newtheorem{example}[thm]{\protect\examplename}
\theoremstyle{definition}
\newtheorem{defn}[thm]{\protect\definitionname{}}
\theoremstyle{plain}
\newtheorem{prop}[thm]{\protect\propositionname}
\theoremstyle{plain}
\newtheorem{lem}[thm]{\protect\lemmaname}
\theoremstyle{remark}
\newtheorem{rem}[thm]{\protect\remarkname}
\DOI{foobar}

\usepackage{babel}
\providecommand\definitionname{Definition}
\providecommand\corollaryname{Corollary}
\providecommand\examplename{Example}
\providecommand\lemmaname{Lemma}
\providecommand\propositionname{Proposition}
\providecommand\remarkname{Remark}
\providecommand\theoremname{Theorem}

\newcommand{\exampleend}{\hfill$\boxempty$}

\usepackage{wasysym}
\usepackage{enumitem}
\setlistdepth{999}

\setlist[itemize]{leftmargin=0.5cm}
\setlist[enumerate]{leftmargin=0.5cm}
\setlist[description]{style=unboxed, leftmargin=0.cm}

\cclogo{\includegraphics{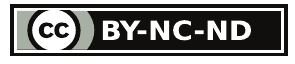}}

\begin{document}

\author[1]{Damien Desfontaines}

\author[2]{Andreas Lochbihler}

\author[3]{David Basin}

\affil[1]{ETH Zurich / Google, damien@desfontain.es}

\affil[2]{Digital Asset, mail@andreas-lochbihler.de}

\affil[3]{ETH Zurich, basin@inf.ethz.ch}

  \title{\huge Cardinality Estimators do not Preserve Privacy}

  \runningtitle{Cardinality Estimators do not Preserve Privacy}

  \begin{abstract}
{Cardinality estimators like HyperLogLog are sketching algorithms that estimate
the number of distinct elements in a large multiset. Their use in
privacy-sensitive contexts raises the question of whether they leak private
information. In particular, can they provide any privacy guarantees while
preserving their strong aggregation properties?\\
We formulate an abstract notion of cardinality estimators, that captures this
aggregation requirement: one can merge sketches without losing precision. We
propose an attacker model and a corresponding privacy definition, strictly
weaker than differential privacy: we assume that the attacker has no prior
knowledge of the data. We then show that if a cardinality estimator satisfies
this definition, then it cannot have a reasonable level of accuracy. We prove
similar results for weaker versions of our definition, and analyze the privacy
of existing algorithms, showing that their average privacy loss is significant,
even for multisets with large cardinalities. We conclude that strong aggregation
requirements are incompatible with any reasonable definition of privacy, and
that cardinality estimators should be considered as sensitive as raw data. We
also propose risk mitigation strategies for their real-world applications.}
  \end{abstract}

  \keywords{differential privacy, cardinality estimators, hyperloglog}

  \journalname{Proceedings on Privacy Enhancing Technologies}
  \DOI{Editor to enter DOI}
  \startpage{1}
  \received{..}
  \revised{..}
  \accepted{..}

  \journalyear{..}
  \journalvolume{..}
  \journalissue{..}

\maketitle

\section{Introduction}

Many data analysis applications must count the number of distinct elements in a
large stream with repetitions. These applications include network
monitoring~\cite{estan2003bitmap}, online analytical
processing~\cite{padmanabhan2003multi,shukla1996storage}, query processing, and
database management~\cite{whang1990linear}. A variety of algorithms, which we
call cardinality estimators, have been developed to do this efficiently, with a
small memory footprint. These algorithms include
PCSA~\cite{flajolet1985probabilistic}, LogLog~\cite{durand2003loglog}, and
HyperLogLog~\cite{flajolet2008hyperloglog}. They can all be parallelized and
implemented using frameworks like MapReduce~\cite{dean2008mapreduce}. Indeed,
their internal memory state, called a \emph{sketch}, can be saved, and sketches
from different data shards can be aggregated without information loss.

Using cardinality estimators, data owners can compute and store sketches over
fine-grained data ranges, for example, daily. Data analysts or tools can then
\emph{merge} (or \emph{aggregate}) existing sketches, which enables the
subsequent estimation of the total number of distinct elements over arbitrary
time ranges. This can be done without re-computing the entire sketch, or even
accessing the original data.

Among cardinality estimators, HyperLogLog~\cite{flajolet2008hyperloglog} and its
variant HyperLogLog++~\cite{heule2013hyperloglog} are widely used in practical
data processing and analysis tasks. Implementations exist for many widely used
frameworks, including Apache Spark~\cite{sparkhll}, Google
BigQuery~\cite{bigqueryhll}, Microsoft SQL Server~\cite{tsqlhll}, and
PostgreSQL~\cite{postgresqlhll}. The data these programs process is
often sensitive. For example, they might estimate the number of distinct IP
addresses that connect to a server~\cite{tschorsch2013algorithm}, the number of
distinct users that perform a particular action~\cite{ashok2014scalable}, or the
number of distinct devices that appeared in a physical
location~\cite{monreale2013privacy}. To illustrate this point, we present an
example of a location-based service, which we will use throughout the paper.

\begin{example}\label{exa:running-example}
A location-based service gathers data about the places visited by the service's
users. For each place and day, the service stores a sketch counting the
identifiers of the users who visited the place that day. This allows the
service's owners to compute useful statistics. For example, a data analyst can
merge the sketches corresponding to the restaurants in a neighborhood over a
month, and estimate how many distinct individuals visited a given restaurant
during that month. The cost of such an analysis is proportional to the number of
aggregated sketches. If the queries used raw data instead, then each query would
require a pass over the entire dataset, which would be much more costly.

Note that the type of analysis to be carried out by the data analyst may not be
known in advance. The data analysts should be able to aggregate arbitrarily many
sketches, across arbitrary dimensions such as time or space. The relative
precision of cardinality estimation should not degrade as sketches are
aggregated.
\exampleend{}
\end{example}

Fine-grained location data is inherently sensitive: it is extreme\-ly
re-identifiable~\cite{golle2009anonymity}, and knowing an individual's location
can reveal private information. For example, it can reveal medical conditions
(from visits to specialized clinics), financial information (from frequent
visits to short-term loan shops), relationships (from regular co-presence),
sexual orientation (from visits to LGBT community spaces), etc. Thus, knowing
where a person has been reveals sensitive information about them. 

As the above example suggests, an organization storing and processing location
data should implement risk mitigation techniques, such as encryption, access
controls, and access audits. The question arises: how should the sketches be
protected? Could they be considered sufficiently aggregated to warrant weaker
security requirements than the raw data?

To answer this question, we model a setting where a data owner stores sketches
for cardinality estimation in a database. The attacker can access some of the
stored sketches and any user statistics published, but not the raw data itself.
This attacker model captures the \emph{insider risk} associated with personal
data collections where insiders of the service provider could gain direct access
to the sketch database. In this paper, we use this insider risk scenario as the
default attacker model when we refer to the attacker. In the discussion, we also
consider a weaker attacker model, modeling an \emph{external attacker} that
accesses sketches via an analytics service provided by the data owner. In both
cases, we assume that the attacker knows the cardinality estimator's internals.

The attacker's goal is to infer whether some user is in the raw data used to
compute one of the accessible sketches. That is, she picks a \emph{target user},
she chooses a sketch built from a stream of users, and she must guess whether
her target is in this stream. The attacker has some \emph{prior knowledge} of
whether her target is in the stream, and examining the sketch gives her a
\emph{posterior knowledge}. The increase from prior to posterior knowledge
determines her knowledge gain.

Consider Example~\ref{exa:running-example}. The attacker could be an employee
trying to determine whether her partner visited a certain restaurant on a given
day, or saw a medical practitioner. The attacker might initially have some
suspicion about this (the prior knowledge), and looking at the sketches might
increase this suspicion. A small, bounded difference of this suspicion might be
deemed to be acceptable, but we do not want the attacker to be able to increase
her knowledge too much.

We show that for all cardinality estimators that satisfy our aggregation
requirement, sketches are almost as sensitive as raw data. Indeed, in this
attacker model, the attacker can gain significant knowledge about the target by
looking at the sketch. Our results are a natural consequence of the aggregation
properties: to aggregate sketches without counting the same user twice, they
must contain information about which users were previously added. The attacker
can use this information, even if she does not know any other users in the
sketch. Furthermore, adding noise either violates the aggregation property, or
has no influence on the success of the attack. Thus, it is pointless to try and
design privacy-preserving cardinality estimators: privacy and accurate
aggregation are fundamentally incompatible.

To show the applicability of our analysis to real-world cardinality estimators,
we quantify the privacy of HyperLogLog, the most widely used cardinality
estimator algorithm. We show that for common parameters, the privacy loss is
significant for most users. Consider, for example, a user $A$ among the 500
users associated with the largest privacy loss, in a sketch that contains 1000
distinct users. An attacker with an initial suspicion of 1\% that $A$ is in the
sketch can raise her level of certainty to over 31\% after observing the sketch.
If her initial estimate is 10\%, it will end up more than 83\%. If her prior is
50\%, then her posterior is as high as 98\%.

Our main contributions are:
\begin{enumerate}
\item We formally define a class of algorithms, which we call \emph{cardinality
estimators}, that count the number of distinct elements in a stream with
repetitions, and can be aggregated arbitrarily many times
(Section~\ref{subsec:def-cardinality-estimators}).

\item We give a definition of privacy that is well-suited to how cardinality
estimators are used in practice (Section~\ref{subsec:privacy-model}).

\item We prove that a cardinality estimator cannot satisfy this privacy property
while maintaining good accuracy asymptotically (Section~\ref{sec:main-result}).
We show similar results for weaker versions of our initial definition
(Section~\ref{sec:weaker-versions}).

\item We highlight the consequences of our results on practical applications of
cardinality estimators (Section~\ref{sec:individual-users}) and we propose risk
mitigation techniques (Section~\ref{sec:discussion}).
\end{enumerate}

\section{Previous work}

Prior work on privacy-preserving cardinality estimators has been primarily
focused on \emph{distributed} user counting, for example to compute user
statistics for anonymity networks like Tor. Each party is usually assumed to
hold a set $X_{i}$ of data, or a sketch built from $X_{i}$, and the goal is to
compute the cardinality of $\bigcup_{i}X_{i}$, without allowing the party $i$ to
get too much information on the sets $X_{j}$ with $j\neq i$. The attacker model
presumes honest-but-curious adversaries.

Tschorsch and Scheuermann~\cite{tschorsch2013algorithm} proposed a noise-adding
mechanism for use in such a distributed context. In~\cite{melis2015efficient},
each party encrypts their sketch, and sends it encrypted to a tally process,
which aggregates them using homomorphic encryption. Ashok et
al.~\cite{ashok2014scalable} propose a multiparty computation protocol based on
Bloom filters to estimate cardinality without the need for homomorphic
encryption, while Egert et al.~\cite{egert2015privately} show that Ashok et
al.'s approach is vulnerable to attacks and propose a more secure variant of the
protocol.

Our attacker model, based on insider risk, is fundamentally different to
previously considered models: the same party is assumed to have access to a
\emph{large number} of sketches; they must be able to aggregate them and get
good estimates.

Our privacy definition for cardinality estimators is inspired from differential
privacy, first proposed by Dwork et al.~\cite{dwork2006calibrating}.
Data-generating probability distributions were considered as a source of
uncertainty
in~\cite{rastogi2009relationship,duan2009privacy,zhou2009differential,bhaskar2011noiseless,kifer2012rigorous,bassily2013coupled,grining2017towards};
and some deterministic algorithms, which do not add noise, have been shown to
preserve privacy under this
assumption~\cite{bhaskar2011noiseless,bassily2013coupled,grining2017towards}.
We explain in Section~\ref{subsec:privacy-model} why we need a custom privacy
definition for our setup and how it relates to differential
privacy~\cite{dwork2008differential} and Pufferfish
privacy~\cite{kifer2012rigorous}.

Our setting has some superficial similarities to the local differential privacy
model, where the data aggregator is assumed to be the attacker. This model is
often used to develop privacy-preserving systems to gather
statistics~\cite{erlingsson2014rappor,fanti2016building,bassily2015local,bassily2017practical}.
However, the setting and constraints of our work differ fundamentally from these
protocols. In local differential privacy, each individual sends data to the
server only once, and there is no need for deduplication or intermediate data
storage. In our work, the need for intermediate sketches and unique counting
leads to the impossibility result.

\section{Cardinality estimators\label{subsec:def-cardinality-estimators}}

In this section, we formally define cardinality estimators, and prove some of
their basic properties.

Cardinality estimators estimate the number of distinct elements in a stream. The
internal state of a cardinality estimator is called a \emph{sketch}. Given a
sketch, one can estimate the number of distinct elements that have been added to
it (the \emph{cardinality}).

Cardinality estimator sketches can also be aggregated: two sketches can be
\emph{merged} to produce another sketch, from which we can estimate the total
number of distinct elements in the given sketches. This aggregation property
makes sketch computation and aggregation embarrassingly parallel. The order and
the number of aggregation steps do not change the final result, so cardinality
estimation can be parallelized using frameworks like MapReduce.

We now formalize the concept of a cardinality estimator. The elements of the
multiset are assumed to belong to a large but finite set $\mathcal{U}$ (the
\emph{universe}).

\begin{defn}\label{defn:deterministic}
A \emph{deterministic cardinality estimator} is a tuple $\left\langle
M_{\varnothing},\mathit{add},\mathit{estimate}\right\rangle $, where
\begin{itemize}
\item $M_{\varnothing}$ is the \emph{empty sketch};
\item $\mathit{add}\left(M,e\right)$ is the deterministic operation that
adds the element $e$ to the sketch $M$ and returns an updated sketch;
\item $\mathit{estimate}\left(M\right)$ estimates the number of distinct
elements that have been added to the sketch.
\end{itemize}
Furthermore, the $\mathit{add}$ operation must satisfy the following axioms for all
sketches $M$ and elements $e, e_1, e_2 \in \mathcal{U}$:
\begin{gather}
  \tag{idempotent}
  \hskip-1.6cm
  \mathit{add}(\mathit{add}(M,e),e)=\mathit{add}(M,e)
  \\
  \tag{commutative}
  \hskip-0.5cm
  \mathit{add}(\mathit{add}(M,e_{1}),e_{2})\!=\!\mathit{add}(\mathit{add}(M,e_{2}),e_{1})
  \hskip-0.2cm
\end{gather}
\end{defn}

These axioms state that $\mathit{add}$ ignores duplicates and that the order in
which elements are added is immaterial. Ignoring duplicates is a natural
requirement for cardinality estimators. Ignoring order is required for this
operation to be used in frameworks like MapReduce, or open-source equivalents
like Hadoop or Apache Beam. Since handling large-scale datasets typically
requires using such frameworks, we consider commutativity to be a hard
requirement for cardinality estimators.

We denote by $M_{e_{1},\dots,e_{n}}$ the sketch obtained by adding
$e_{1},\dots,e_{n}$ successively to $M_{\varnothing}$, and we denote by
$\mathcal{M}$ the set of all sketches that can be obtained inductively from
$M_{\varnothing}$ by adding elements from any subset of $\mathcal{U}$ in any
order. Note that $\mathcal{M}$ is finite and of cardinality at most
$2^{\left|\mathcal{U}\right|}$. Order and multiplicity do not influence
sketches: we denote by $M_{E}$ the sketch obtained by adding all elements of a
set $E$ (in any order) to $M_{\varnothing}$.

Later, in Example~\ref{exa:cardinality-estimators}, we give several examples of
deterministic cardinality estimators, all of which satisfy
Definition~\ref{defn:deterministic}.

\begin{lem}\label{ce-property}
  $\mathit{add}\left(M_{e_{1},\dots,e_{n}},e_{i}\right)=M_{e_{1},\dots,e_{n}}$ for all $i\leq n$.
\end{lem}

\begin{proof}
This follows directly from Properties 1 and 2.
\end{proof}

In practice, cardinality estimators also have a $\mathit{merge}$ operation. We
do not explicitly require the existence of this operation, since
$\mathit{add}$'s idempotence and commutativity ensure its existence as follows.

\begin{defn}\label{defn:merge}
To merge two sketches $M$ and $M^{\prime}$, choose some $E = \left\{ e_{1},\dots,e_{n}\right\} \subseteq \mathcal{U}$ such that $M^{\prime}=M_{E}$.
We define $\mathit{merge}(M,M^{\prime})$ to be the sketch obtained after adding all elements of $E$ successively to $M$.
\end{defn}

Note that this construction is not computationally tractable, even though in
practical scenarios, the $\mathit{merge}$ operation must be fast. This
efficiency requirement is not necessary for any of our results, so we do not
explicitly require it either.

\begin{lem}\label{lem:merge-properties}
The $\mathit{merge}$ operation in Definition~\ref{defn:merge} is well-defined,
i.e., it does not depend on the choice of $E$. Furthermore, the merge operation
is a commutative, idempotent monoid on $\mathcal{M}$ with $M_{\varnothing}$ as
neutral element.
\end{lem}

The proof is given in Appendix~\ref{lem:merge-properties-proof}. Note that these
properties of the merge operation are important for cardinality estimators: when
aggregating different sketches, we must ensure that the result is the same no
matter in which order the sketches are aggregated.

Existing cardinality estimators also satisfy efficiency requirements: they have
a low memory footprint, and $\mathit{add}$ and $\mathit{merge}$ run in constant
time. These additional properties are not needed for our results, so we omit
them in our definition.

We now define \emph{precise} cardinality estimators.

\begin{defn}
Let $E$ be a set of cardinality $n$ taken uniformly at random in
$\mathcal{U}$. The quality of a cardinality estimation algorithm
is given by two metrics:

\begin{enumerate}
\item Its \emph{bias} $\mathbb{E}\left[n-\mathit{estimate}\left(M_{E}\right)\right]$
\item Its \emph{variance}
\[
  \mathbb{V}_{n}=\mathbb{E}\left[{\left(\mathbb{E}\left[\mathit{estimate}\left(M_E\right)\right]-\mathit{estimate}\left(M_{E}\right)\right)}^{2}\right].
\]
\end{enumerate}

Cardinality estimators used in practice are asymptotically unbiased:
$\frac{1}{n}\cdot\mathbb{E}\left[n-\mathit{estimate}\left(M_{E}\right)\right]=o\left(1\right)$.
In the rest of this work, we assume that all cardinality estimators we consider
are perfectly unbiased, so
  $\mathbb{V}_{n}=\mathbb{E}\left[{\left(n-\mathit{estimate}\left(M_{E}\right)\right)}^{2}\right]$.
Cardinality estimators are often compared by their relative standard error
(RSE), which is given by
  $\sqrt{\mathbb{E}\left[{\left(\frac{n-\mathit{estimate}\left(M_{E}\right)}{n}\right)}^{2}\right]}=\frac{\sqrt{\mathbb{V}_{n}}}{n}$.

A cardinality estimator is said to be \emph{precise} if it is asymptotically
unbiased and its relative standard error is bounded by a constant. In practice,
we want the relative standard error to be less than a reasonably small constant,
for example less than 10\%.
\end{defn}

\begin{example}\label{exa:cardinality-estimators}

We give a few examples of cardinality estimators, with their \emph{memory usage}
in bits ($m=\log_2\left|\mathcal{M}\right|$) and their relative standard error.
As a first step, they all apply a hash function to the user identifiers.
Conceptually, this step assigns to all users probabilistically a random
bitstring of length 32 or 64 bits. Since hashing is deterministic, all
occurrences of a user are mapped to the same bitstring.

\begin{itemize}
\item K-Minimum Values~\cite{bar2002counting}, with parameter $k$, maintains a list of the $k$ smallest
hashes that have been added to the sketch. With $32$-bit hashes, it has a memory
usage of $m=32\cdot k$, and its unbiased estimator has a
RSE of approximately
$\frac{1}{\sqrt{k}}\simeq\frac{5.66}{\sqrt{m}}$~\cite{beyer2007synopses}.
\item Probabilistic Counting with Stochastic Averaging (PCSA, also known as
FM-sketches)~\cite{flajolet1985probabilistic} maintains a list of $k$ \emph{bit arrays}. When an element is added
to an FM-sketch, its hash is split into two parts. The first part determines
which bit array is modified. The second part determines which bit is flipped to
$1$, depending on the number of consecutive zeroes at the beginning. With $k$
registers and $32$-bit hashes, its memory usage is $m=32\cdot k$, and its
RSE is approximately
$\frac{0.78}{\sqrt{k}}\simeq\frac{4.41}{\sqrt{m}}$.
\item LogLog~\cite{durand2003loglog} maintains a tuple of $k$ registers. Like PCSA, it uses the first
part of the hash to pick which register to modify. Then, each registers stores
the \emph{maximum} number of consecutive zeroes observed so far. Using $k\geq64$
registers and a $32$-bit hash, its memory usage is $m=5\cdot k$ bits, and its
standard error is approximately
$\frac{1.30}{\sqrt{k}}\simeq\frac{2.91}{\sqrt{m}}$.
\item HyperLogLog~\cite{heule2013hyperloglog} has the same structure and add operation as LogLog, only its
estimate operation is different. Using $k\geq16$ registers and a $64$-bit hash,
it has a memory usage of $m=6\cdot k$ bits, and its standard error is
approximately
$\frac{1.04}{\sqrt{k}}\simeq\frac{2.55}{\sqrt{m}}$.
\item Bloom filters can also be used for cardinality
estimation~\cite{papapetrou2010cardinality}. However, we could not find an
expression of the standard error for a given memory usage in the literature.
\end{itemize}

All these cardinality estimators have null or negligible ($\ll 1$)
bias. Thus, their variance is equal to their mean squared standard error. So the
first four are precise, whereas we do not know if Bloom filters are.
\exampleend{}
\end{example}

All examples above are deterministic cardinality estimators. For them and other
deterministic cardinality estimators, bias and variance only come from the
randomness in the algorithm's inputs. We now define \emph{probabilistic}
cardinality estimators. Intuitively, these are algorithms that retain all the
useful properties of deterministic cardinality estimators, but may flip coins
during computation. We denote by $\mathfrak{M}$ the set of distributions over
$\mathcal{M}$.

\begin{defn}\label{defn:probabilistic}
A \emph{probabilistic cardinality estimator} is a tuple $\left\langle M_{\varnothing},\mathit{add},\mathit{merge},\mathit{estimate}\right\rangle $,
where
\begin{itemize}
\item $M_{\varnothing}\in\mathcal{M}$ is the \emph{empty sketch};
\item $\mathit{add}\left(M,e\right):\mathcal{M}\times\mathcal{U}\longrightarrow\mathfrak{M}$
is the probabilitistic operation that adds the element $e$ to the sketch $M$ and
returns an updated sketch;
\item $\mathit{merge\left(M_1,M_2\right)}:\mathcal{M}\times\mathcal{M}\longrightarrow\mathfrak{M}$
is the probabilistic operation that merges two sketches $M_1$ and $M_2$; and
\item $\mathit{estimate}\left(M\right):\mathcal{M}\longrightarrow\mathbb{N}$
estimates the number of unique elements that have been added to the sketch.
\end{itemize}

Both the $\mathit{add}$ and $\mathit{merge}$ operations can be extended to
\emph{distributions} of sketches. For a distribution of sketches $\mathcal{D}$
and an element $e$, $\mathit{add}\left(\mathcal{D},e\right)$ denotes the
distribution such that:
\[
 \mathbb{P}\left[\mathit{add}\left(\mathcal{D},e\right)=M_0\right]=\sum_{M}\mathbb{P}\left[\mathcal{D}\!=\!M\right]\mathbb{P}\left[\mathit{add}\left(M,e\right)\!=\!M_0\right].
\]
For two distributions of sketches $\mathcal{D}$ and $\mathcal{D}^{\prime}$,
$\mathit{merge}\left(\mathcal{D},\mathcal{D}^{\prime}\right)$ denotes the
distribution such that:
\begin{align*}
 & \mathbb{P}\left[\mathit{merge}\left(\mathcal{D},\mathcal{D}^{\prime}\right)=M_0\right]\\
 & =\sum_{M,M^\prime}\mathbb{P}\left[\mathcal{D}\!=\!M\right]\mathbb{P}\left[\mathcal{D^\prime}\!=\!M^\prime\right]\mathbb{P}\left[\mathit{merge}\left(M,M^\prime\right)\!=\!M_0\right].
\end{align*}

We want probabilistic cardinality estimators to have the same high-level
properties as deterministic cardinality estimators: idempotence, commutativity,
and the existence of a well-behaved $\mathit{merge}$ operation. In the
deterministic case, the idempotence and commutativity of the $\mathit{add}$
operation was sufficient to show the existence of a $\mathit{merge}$ operation
with the desired properties. In the probabilistic case, this no longer holds.
Instead, we require the following two properties.

\begin{itemize}
\item For a set $E\subseteq\mathcal{U}$, let $\mathcal{D}_{E}$ denote the sketch
distribution obtained when adding elements of $E$ successively into
$M_{\varnothing}$. The mapping from $E$ to $\mathcal{D}_{E}$ must be
well-defined: it must be independent of the order in which we add elements, and
possible repetitions. This requirement encompasses both idempotence and
commutativity.
\item For two subsets $E_{1}$ and $E_{2}$ of $\mathcal{U}$, we require that\\
$\mathit{merge}\left(\mathcal{D}_{E_{1}},\mathcal{D}_{E_{2}}\right)=\mathcal{D}_{E_{1}\cup E_{2}}$.
\end{itemize}
These requirements encompass the results of Lemma~\ref{lem:merge-properties}.
\end{defn}

These properties, like in the deterministic case, are very strong. They impose
that an arbitrary number of sketches can be aggregated without losing accuracy
during the aggregation process. This requirement is however realistic in many
practical contexts, where the same sketches are used for fine-grained analysis
and for large-scale cardinality estimation. If this requirement is relaxed, and
the cardinality estimator is allowed to return imprecise results when merging
sketches, our negative results do not hold.

For example, Tschorsch and Scheuermann proposed a cardinality estimation
scheme~\cite{tschorsch2013algorithm} which adds noise to sketches to make them
satisfy privacy guarantees in a distributed context. However, their algorithm is
not a probabilistic cardinality estimator according to our definition: noisy
sketches can no longer be aggregated. Indeed,~\cite{tschorsch2013algorithm}
explains that ``combining many perturbed sketches quickly drives [noise] to
exceedingly high values.'' In our setting, aggregation is crucial, so we do not
further consider their algorithm.

\section{Modeling privacy}

In this section, we describe our system and attacker model
(\ref{subsec:system:attacker}), and present our privacy
definition~(\ref{subsec:privacy-model}).

\subsection{System and attacker model\label{subsec:system:attacker}}

\tikzset{
  user/.pic={
    \draw{}
    (0, 0.35cm) circle [radius=0.1cm]
    ++(0, -0.1cm) -- ++(0, -0.3cm)
    -- +(-110:0.3cm)
    ++(0, 0) -- +(-70:0.3cm)
    ++(0, 0) ++(-0.15cm, 0.2cm) -- ++(0.3cm, 0);
  }
}

\begin{figure}
  \begin{tikzpicture}[node distance=0.2cm and 0.5cm, align=center, font=\small, scale=0.7]
    \node (userdata) {
      user
      \\
      data
    };

    \node (rawdb) [above right=of userdata, draw, rounded corners] {
      raw user
      \\
      database
    };
    
    \node (sketching) [below right=of userdata, draw, rectangle] {
      sketching
      \\
      algorithm
    };

    \node (sketchdb) [right=of sketching, draw, rounded corners] {
      sketch
      \\
      database
    };

    \node (aggregation) at ($ (sketchdb)!0.5!(rawdb) $) [draw, rectangle] {
      aggregation%
    };

    \node (estimator) at ([xshift=0.3cm] sketchdb.west |- rawdb) [anchor=west, draw, rectangle] {
      estimation%
    };

    \node (statistics) [right=of estimator, draw, ellipse, inner sep=0.15em] {
      published
      \\
      statistics
    };

    \node (analytics) at (statistics |- sketchdb) [draw, ellipse, inner sep=0.15em] {
      analytics
      \\
      service
    };

    \coordinate (user1) at ([xshift=-0.8cm, yshift=1cm] userdata.west);
    \coordinate (user2) at ([xshift=-1cm, yshift=0cm] userdata.west);
    \coordinate (user3) at ([xshift=-0.8cm, yshift=-1cm] userdata.west);
    \draw
      (user1) pic{user}
      (user2) pic{user}
      (user3) pic{user};
    \begin{scope}[dashed]
      \draw [<-] (userdata.135) -- +(135:0.8cm);
      \draw [<-] (userdata.180) -- +(180:0.8cm);
      \draw [<-] (userdata.225) -- +(225:0.8cm);      
    \end{scope}

    \draw[->] (userdata) -- (rawdb);
    \draw[->] (rawdb) -- (sketching.north -| rawdb.south);
    \draw[->] (userdata) -- (sketching);
    \draw[->] (sketching) -- (sketchdb);
    \draw[->] (sketchdb) to[out=160, in=-90] (aggregation) (aggregation) to[out=0,in=90] (sketchdb);
    \draw[->] (sketchdb.70) -- (estimator);
    \draw[->] (estimator) -- (statistics);
    \draw[->] (sketchdb) -- (analytics);

    \begin{pgfonlayer}{background}
      \node (serviceprovider) [fit={(rawdb) (sketching) (sketchdb) (estimator)}, draw, fill=black!20, densely dotted] {};
    \end{pgfonlayer}
    \node at (serviceprovider.north) [anchor=south] {service provider};

    \node (attacker) at (serviceprovider.south) [anchor=north, yshift=-1cm, rectangle, fill=black, text=white] {
      attacker%
    };
    \draw[->] [overlay] (attacker) to[out=180, in=-90, looseness=0.8] 
      node [pos=0.4, anchor=south, sloped] {controls} 
      ([yshift=-0.5cm] user3);
    \draw[->] [overlay] (attacker) to[out=10, in=-90, looseness=0.8] 
      node [pos=0.4, anchor=south, sloped] {uses}
      (analytics);
    \draw[->] (attacker) -- 
      node [pos=0.4, anchor=east, inner xsep=0pt] {internal}
      node [pos=0.4, anchor=west, inner xsep=2pt] {attack}
      (sketchdb);
    \draw[->] [overlay] (attacker) to[out=0, in=-50, looseness=1.65]
      node [pos=0.25, anchor=north, sloped] {reads}
      (statistics);
  \end{tikzpicture}
  \caption{System and attacker model}\label{figure:system:model}
\end{figure}

\noindent
Figure~\ref{figure:system:model} shows our system and attacker model. The
service provider collects sensitive data from many users over a long time span.
The raw data is stored in a database. Over shorter time periods (e.g.\ an hour,
a day, or a week), a cardinality estimator aggregates all data into a sketch.
Sketches of all time periods are stored in a sketch database. Sketches from
different times are aggregated into sketches of longer time spans, which are
also stored in the database. Estimators compute user statistics from the
sketches in the database, which are published. The service provider may also
publish the sketches via an analytics service for other parties.

The attacker knows all algorithms used (those for sketching, aggregation, and
estimation, including their configuration parameters such as the hash function
and the number of buckets) and has access to the published statistics and the
analytics service. She controls a small fraction of the users that produce user
data. However, she can neither observe nor change the data of the other users.
She also does not have access to the database containing the raw data.

In this work, we mainly consider an \emph{internal attacker} who has access to
the sketch database. For this internal attacker, the goal is to discover whether
her target belongs to a given sketch. We then discuss how our results extend to
weaker \emph{external attackers}, which can only use the analytics service. We
will see that for our main results, the attacker only requires access to one
sketch. The possibility to use multiple sketches will only come up when
discussing mitigations strategies in Section~\ref{subsec:salting-the-hash}.

\subsection{Privacy definition\label{subsec:privacy-model}}

\noindent
We now present the privacy definition used in our main result. Given the system
and attacker just described, our definition captures the impossibility for the
attacker to gain significant positive knowledge about a given target. We explain
this notion of knowledge gain, state assumptions on the attacker's prior
knowledge, and compare our privacy notion with other well-known definitions.

We define a very \emph{weak} privacy requirement: reasonable definitions used
for practical algorithms would likely be stronger. Working with a weak
definition \emph{strengthens} our negative result: if a cardinality estimator
satisfying our weak privacy definition cannot be precise, then this is also the
case for cardinality estimators satisfying a stronger definition.

In Section~\ref{sec:weaker-versions}, we explore even weaker privacy definitions
and prove similar negative results (although with a looser bound). In
Section~\ref{sec:individual-users}, we relax the requirement that \emph{every
individual user} must be protected according to the privacy definition. We show
then that our theorem no longer holds, but that practical uses of cardinality
estimators still cannot be considered privacy-preserving.

We model a possible attack as follows. The attacker has access to the identifier
of a user $t$ (her \emph{target}) and a sketch $M$ generated from a set of
users $E$ ($M\leftarrow M_{E}$) unknown to her. The attacker wants to know
whether $t\in E$. She initially has a \emph{prior} knowledge of whether the
target is in the sketch. Like in Bayesian inference, $\frac{\mathbb{P}\left[t\in
E\right]}{\mathbb{P}\left[t\notin E\right]}$ represents how much more likely the
user is in the database than is not, according to the attacker. After looking at
the sketch $M$, this knowledge changes: her \emph{posterior} knowledge becomes
$\frac{\mathbb{P}\left[t\in E \mid M_{E}=M\right]}{\mathbb{P}\left[t\notin E
\mid M_E=M\right]}$.

We define privacy to capture that the attacker's posterior knowledge should not
increase too much. In other words, the attacker should not gain significant
knowledge by seeing the sketch. This must hold for every possible sketch $M$ and
every possible user $t$. We show in Lemma~\ref{lem:equivalent-definition} that
the following definition bounds the \emph{positive} knowledge gain of the
attacker.

\begin{defn}\label{def:epsilon-sketch-privacy}
A cardinality estimator satisfies \emph{$\varepsilon$-sketch privacy above
cardinality $N$} if for every $n\geq N$, $t\in\mathcal{U}$, and
$M\in\mathcal{M}$, the following inequality holds:
\[
\mathbb{P}_{n}\left[M_E=M \mid t\in E\right]\leq
e^{\varepsilon}\cdot\mathbb{P}_{n}\left[M_E=M \mid t\notin E\right].
\]
Here, the probability $\mathbb{P}_{n}$ is taken over:
\begin{itemize}
\item a uniformly chosen set $E\in\mathcal{P}_{n}\left(\mathcal{U}\right)$,
where $\mathcal{P}_n\left(\mathcal{U}\right)$ is the set of all possible subsets
$E\subseteq\mathcal{U}$ of cardinality $n$; and
\item the coin flips of the algorithm, for probabilistic cardinality estimators.
\end{itemize}
\end{defn}

If a cardinality estimator satisfies this definition, then for any user $t$, the
probability of observing $M$ if $t\in E$ is not much higher than the probability
of observing $M$ if $t\notin E$. To give additional intuition on
Definition~\ref{def:epsilon-sketch-privacy}, we now show that the parameter
$\varepsilon$ effectively captures the attacker's positive knowledge gain.

\begin{lem}\label{lem:equivalent-definition}
A cardinality estimator satisfies $\varepsilon$-sketch privacy above cardinality
$N$ if and only if the following inequality holds for every $n\geq N$,
$t\in\mathcal{U}$ and $M\in\mathcal{M}$ with
$\mathbb{P}_{n}\left[M_E=M\right]>0$:
\[
\frac{\mathbb{P}_{n}\left[t\in E \mid M_E=M\right]}{\mathbb{P}_{n}\left[t\notin E \mid M_E=M\right]}
\leq e^{\varepsilon}\cdot\frac{\mathbb{P}_{n}\left[t\in E\right]}{\mathbb{P}_{n}\left[t\notin E\right]}.
\]
\end{lem}
\begin{proof}
Bayes' law can be used to derive one inequality from the other.
We have
\begin{gather*}
  \mathbb{P}_{n}[M_E\!=\!M \mid t \in E] =
  \frac{\mathbb{P}_{n}[t\in E \mid M_E\!=\!M]\cdot\mathbb{P}_{n}[M_E\!=\!M]}{\mathbb{P}_{n}[t\in E]}
  \\
  \mathbb{P}_{n}[M_E\!=\!M \mid t\notin E]=\frac{\mathbb{P}_{n}[t\notin E \mid M_E\!=\!M]\cdot\mathbb{P}_{n}[M_E\!=\!M]}{\mathbb{P}_{n}[t\notin E]}.
\end{gather*}
The equivalence between the definitions follows directly.
\end{proof}

This definition has three characteristics which make it unusually weak. They
correspond to an \emph{under\hyp{}approximation} of the attacker's capabilities
and goals.

\begin{description}

\item[Uniform prior] The choice of distribution for $\mathbb{P}_{n}$ implies
that the elements of $E$ are \emph{uniformly distributed} in $\mathcal{U}$. This
corresponds to an attacker who has \emph{no prior knowledge} about the data. In
the absence of prior information about the elements of the set $E$, the
attacker's best approximation is the uniform distribution. In practice, a
realistic attacker might have more information about the data, so a stronger
privacy definition would model this prior knowledge by a larger family of
probability distributions. More precisely, since the elements of $E$ are
uniformly distributed in $\mathcal{U}$, the prior knowledge from the attacker
$\mathbb{P}_n\left[t\in E\right]$ is exactly
$\left|E\right|/\left|\mathcal{U}\right|$. A realistic attacker would likely
have a larger prior knowledge about their target. However, any reasonable
definition of privacy would also include the case where the attacker does not
have more information on their target than on other users and, as such, would be
stronger than $\epsilon$-sketch privacy.

\item[Asymmetry] We only consider the \emph{positive} information gain by the
attacker. There is an \emph{upper bound} on the probability that $t\in E$ given
the observation $M$, but no lower bound. In other words, the attacker is allowed
to deduce with absolute certainty that $t\notin E$. In practice, both positive
\emph{and} negative information gains may present a privacy risk. In our running
example (see Example~\ref{exa:running-example}), deducing that a user did
\emph{not} spend the night at his apartment could be problematic.

\item[Minimum cardinality] We only require a bound on the information gain for
cardinalities larger than a parameter $N$. In practice, $N$ could represent a
threshold over which it is considered safe to publish sketches or to relax data
protection requirements. Choosing a small $N$ (like $N=10$) strengthens the 
privacy definition, while choosing a large $N$ (like $N=500$)
limits the utility of the data, as many smaller sketches cannot be published.
\end{description}

We emphasize again that these characteristics, which result in a very weak
definition, make our notion of privacy well-suited to proving \emph{negative
results}. If satisfying our definition is impossible for an accurate cardinality
estimator, then a stronger definition would similarly be impossible to satisfy.
For example, any reasonable choice of distributions used to represent the prior
knowledge of the attacker would include the uniform distribution.

We now compare our definition to two other notions: differential
privacy~\cite{dwork2008differential} and Pufferfish
privacy~\cite{kifer2012rigorous}.

\subsection{Relation to differential privacy}

\noindent
Recall the definition of differential privacy: $\mathcal{A}$ is
$\varepsilon$-differentially private if and only if
$e^{-\varepsilon}\leq\frac{\mathbb{P}\left[\mathcal{A}\left(D_{1}\right)\right]}{\mathbb{P}\left[\mathcal{A}\left(D_{2}\right)\right]}\leq e^{\varepsilon}$
for any databases $D_{1}$ and $D_{2}$ that only differ by one element. In our
setup, this could be written as the two inequalities $\mathbb{P}\left[M_{E} \mid
t\in E\right]\leq e^{\varepsilon}\cdot\mathbb{P}\left[M_{E} \mid t\notin
E\right]$ and $\mathbb{P}\left[M_{E} \mid t\notin E\right]\leq
e^{\varepsilon}\cdot\mathbb{P}\left[M_{E} \mid t\in E\right]$.

Asymmetry, and minimum cardinality, are two obvious differences between our
notion of privacy and differential privacy. But the major difference lies in the
\emph{source of uncertainty}. In differential privacy, the probabilities are
taken over the coin flips of the algorithm. The attacker is implicitly assumed
to know the algorithm's input except for \emph{one} user: the uncertainty comes
entirely from the algorithm's randomness. In our definition, the attacker has no
prior knowledge of the input, so the uncertainty comes either entirely from the
attacker's lack of background knowledge (for deterministic cardinality
estimators), or both from the attacker's lack of background knowledge and the
algorithm's inherent randomness.

The notion of relying on the initial lack of knowledge of the attacker in a
privacy definition is not new: it is for example made explicit in the definition
of Pufferfish privacy, a generic framework for privacy definitions.

\subsection{Relation to Pufferfish privacy\label{subsec:relation-to-pufferfish}}

\noindent
Pufferfish privacy~\cite{kifer2012rigorous} is a customizable framework for
building privacy definitions. A Pufferfish privacy definition has three
components: a set of \emph{potential secrets} $\mathbb{S}$; a set of
\emph{discriminative pairs}
$\mathbb{S}_{\text{pairs}}\subseteq\mathbb{S}\times\mathbb{S}$; and
a set of \emph{data evolution scenarios} $\mathbb{D}$.

$\mathbb{S}_{\text{pairs}}$ represents the facts we want the attacker to be
unable to distinguish. In our case, we want to prevent the attacker from
distinguishing between $t\in E$ and $t\notin E$:
$\mathbb{S}_{\text{pairs}}=\left\{ \left(t\in E,t\notin E\right) \mid
t\in\mathcal{U}\right\}$. $\mathbb{D}$ represents what the possible
distributions of the input data are. In our case, it is a singleton that only
contains the uniform distribution.

Our definition is \emph{almost} an instance of Pufferfish privacy. Like with
differential privacy, the main difference is asymmetry.

The close link to Pufferfish privacy supports our proof of two fundamental
properties of privacy definitions: \emph{transformation invariance} and
\emph{convexity}~\cite{kifer2012rigorous}. Transformation invariance states that
performing additional analysis of the output of the algorithm does not allow an
attacker to gain more information, i.e., the privacy definition is closed under
composition with probabilistic algorithms. Convexity states that if a data owner
chooses randomly between two algorithms satisfying a privacy definition and
generates the corresponding output, this procedure itself will satisfy the same
privacy definition.
These two properties act as \emph{sanity checks} for our privacy definition.

\begin{prop}
$\varepsilon$-sketch privacy above cardinality $N$ satisfies transformation
invariance and convexity.
\end{prop}
\begin{proof}
The proof is similar to the proof of Theorem 5.1 in~\cite{kifer2012rigorous},
proved in Appendix B of the same paper.
\end{proof}

\section{Private cardinality estimators are imprecise\label{sec:main-result}}

Let us return to our privacy problem: someone with access to a sketch wants to
know whether a given individual belongs to the aggregated individuals in the
sketch. Formally, given a target $t$ and a sketch $M_{E}$, the attacker must
guess whether $t\in E$ with high probability. In
Section~\ref{subsec:main-result-deterministic}, we explain how the attacker can
use a simple test to gain significant information if the cardinality estimator
is deterministic. Then, in Section~\ref{subsec:main-result-probabilistic}, we
reformulate the main technical lemma in probabilistic terms, and prove an
equivalent theorem for probabilistic cardinality estimators.

\subsection{Deterministic case\label{subsec:main-result-deterministic}}

Given a target $t$ and a sketch $M_E$, the attacker can perform the following
simple attack to guess whether $t \in E$. She can try to add the target $t$ to
the sketch $M_{E}$, and observe whether the sketch changes. In other words, she
checks whether $\mathit{add}\left(M_{E},t\right)=M_{E}$. If the sketch changes,
this means with certainty that $t \notin E$. Thus, Bayes' law indicates that if
$\mathit{add}\left(M_{E},t\right)=M_{E}$, then the probability of $t\in E$
cannot decrease.

How large is this increase? Intuitively, it depends on how likely it is that
adding an element to a sketch does not change it \emph{if the element has not
previously been added to the sketch}. Formally, it depends on
$\mathbb{P}\left[\mathit{add}\left(M_{E},t\right)=M_{E} \mid t\notin E\right]$.

\begin{itemize}
\item If $\mathbb{P}\left[\mathit{add}\left(M_{E},t\right)=M_{E} \mid t\notin E\right]$
is close to $0$, for example if the sketch is a list of all elements seen
so far, then observing that $\mathit{add}\left(M_{E},t\right)=M_{E}$ will lead
the attacker to believe with high probability that $t\in E$.
\item If $\mathbb{P}\left[\mathit{add}\left(M_{E},t\right)=M_{E} \mid t\notin E\right]$
is close to $1$, it means that adding an element to a sketch often
does not change it. The previous attack does not reveal much information.
But then, it also means that many elements are ignored when they are added
to the sketch, that is, the sketch does not change when
adding the element. Intuitively, the accuracy of an estimator based solely on a
sketch that ignores many elements cannot be very good.
\end{itemize}
We formalize this intuition in the following theorem.

\begin{figure*}
\includegraphics[height=4cm]{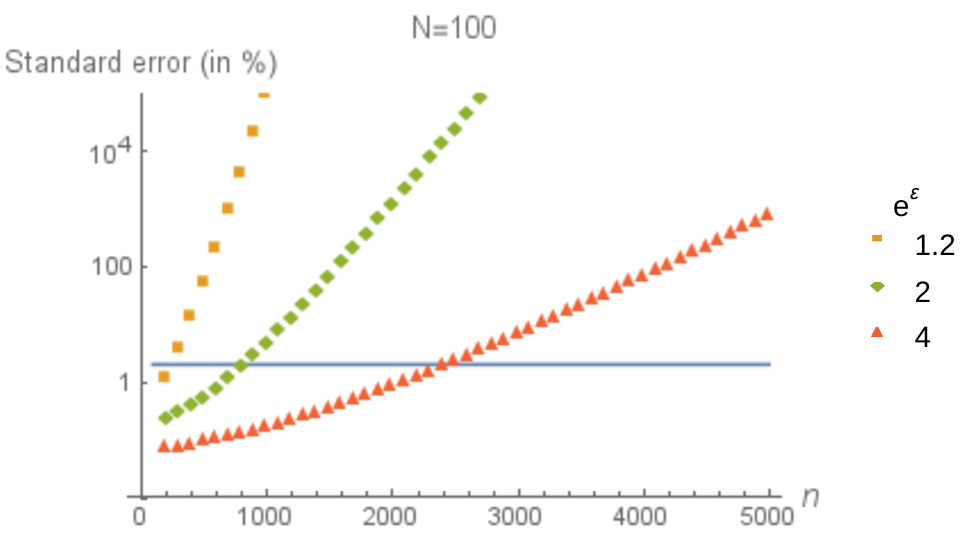}
\hfill
\includegraphics[height=4cm]{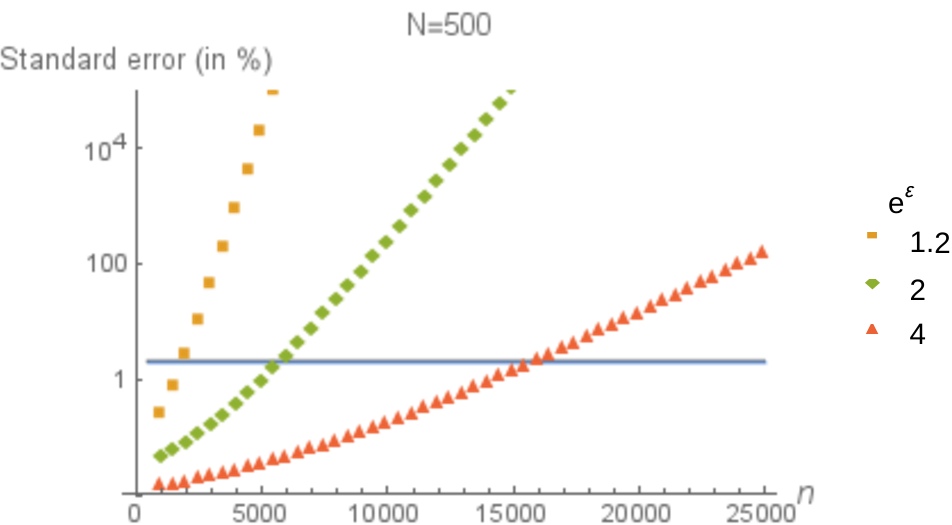}
\caption{\label{fig:minstderr}Minimum standard error for a cardinality estimator
with $\varepsilon$-sketch privacy above cardinality $100$ (left) and $500$ (right).
The blue line is the relative standard error of HyperLogLog with standard parameters.}
\end{figure*}

\begin{thm}\label{thm:main-result}
An unbiased deterministic cardinality estimator that satisfies
$\varepsilon$-sketch privacy above cardinality $N$ is not precise. Namely, its
variance is at least $\frac{1-c^{k}}{c^{k}}\left(n-k\cdot N\right)$, for any
$n\geq N$ and $k\leq\frac{n}{N}$, where $c = 1-e^{-\varepsilon}$
\end{thm}

Note that if we were using differential privacy, this result would be trivial:
no deterministic algorithm can ever be differentially private. However, this is
not so obvious for our definition of privacy: prior
work~\cite{bhaskar2011noiseless,bassily2013coupled,grining2017towards} shows
that when the attacker is assumed to have some uncertainty about the data, even
deterministic algorithms can satisfy the corresponding definition of privacy.

Figure~\ref{fig:minstderr} shows plots of the lower bound on the standard error of
a cardinality estimator with $\varepsilon$-sketch privacy at two
cardinalities (100 and 500). It shows that the standard error increases
exponentially with the number of elements added to the sketch. This demonstrates that
even if we require the privacy property for a large value of $N$ (500) and a
large $\varepsilon$ (which is generally less than $1$), the standard error of a
cardinality estimator will become unreasonably large after 20,000 elements.

\vspace{\topsep}
\begin{proof}[Proof of Theorem~\ref{thm:main-result}]
The proof is comprised of three steps, following the intuition previously given.
\begin{enumerate}
\item We show that a sketch $M_{E}$, computed from a random set
$E$ with an $\varepsilon$-sketch private estimator above
cardinality $N$, will ignore many elements after $N$ (Lemma~\ref{lem:ignores-elements}).
\item We prove that if a cardinality estimator ignores a certain ratio of
elements after adding $n=N$ elements, then it will ignore an even larger ratio
of elements as $n$ increases (Lemma~\ref{lem:ignores-more-and-more-elements}).
\item We conclude by proving that an unbiased cardinality estimator that
ignores many elements must have a large variance (Lemma~\ref{lem:bad-variance}).
\end{enumerate}
The theorem follows directly from these lemmas.
\end{proof}

\begin{lem}\label{lem:ignores-elements}
Let $t\in\mathcal{U}$. A deterministic cardinality estimator with
$\varepsilon$-sketch privacy above cardinality $N$ satisfies
$\mathbb{P}_{n}\left[\mathit{add}\left(M_{E},t\right)=M_{E} \mid t\notin E\right]\geq e^{-\varepsilon}$
for $n\geq N$.
\end{lem}
\begin{proof}
We first prove that such an estimator satisfies
\begin{multline*}
  \mathbb{P}_{n}\left[\mathit{add}\left(M_{E},t\right)=M_{E} \mid t\in E\right]{}
  \\
  \leq e^{\varepsilon}\cdot\mathbb{P}_{n}\left[\mathit{add}\left(M_{E},t\right)=M_{E} \mid t\notin E\right].
\end{multline*}
We decompose the left-hand side of the inequality over all possible values of
$M_{E}$ which that $\mathit{add}\left(M_{E},t\right)=M_{E}$. If we call
this set $\mathcal{I}_{t}=\left\{ M \mid
\mathit{add}\left(M,t\right)=M\right\}$, we have:
\begin{align*}
  \mathbb{P}_{n}&\left[\mathit{add}\left(M_{E},t\right)=M_{E} \mid t\in E\right]
  \\
  & {} = \sum\nolimits_{M\in\mathcal{I}_{t}}\mathbb{P}_{n}\left[M_E=M \mid t\in E\right]
  \\
  & {} \leq e^{\varepsilon}\cdot\sum\nolimits_{M\in\mathcal{I}_{t}}\mathbb{P}_{n}\left[M_E=M \mid t\notin E\right]
  \\
  & {} \leq e^{\varepsilon}\cdot\mathbb{P}_{n}\left[\mathit{add}\left(M_{E},t\right)=M_{E} \mid t\notin E\right],
\end{align*}
where the first inequality is obtained directly from the definition of
$\varepsilon$-sketch privacy.

Now, Lemma~\ref{ce-property} gives $\mathbb{P}\left[\mathit{add}\left(M_{E},t\right)=M_{E} \mid t\in E\right]=1$,
and finally $\mathbb{P}_{n}\left[\mathit{add}\left(M_{E},t\right)=M_{E} \mid t\notin E\right]\geq e^{-\varepsilon}$.
\end{proof}

\begin{lem}\label{lem:ignores-more-and-more-elements}
Let $t\in\mathcal{U}$. Suppose a deterministic cardinality estimator satisfies
$\mathbb{P}_{n}\left[\mathit{add}\left(M_{E},t\right)=M_{E} \mid t\notin E\right]\geq p$
for any $n\geq N$. Then for any integer $k\geq1$, it also satisfies
$\mathbb{P}_{n}\left[\mathit{add}\left(M_{E},t\right)=M_{E} \mid t\notin E\right]\geq1-{\left(1-p\right)}^{k}$,
for $n\geq k\cdot N$.
\end{lem}
\begin{proof}
First, note that if $F\subseteq E$, and
$\mathit{add}\left(M_{F},t\right)=M_{F}$, then
$\mathit{add}\left(M_{E},t\right)=M_{E}$. This is a direct consequence of
Lemma~\ref{lem:merge-properties}:
$M_{E}=\mathit{merge}\left(M_{E\backslash F},M_{F}\right)$, so:
\begin{align*}
  \mathit{add}\left(M_{E},t\right) & {} = \mathit{merge}\left(M_{E\backslash F},\mathit{add}\left(M_{F},t\right)\right)
  \\
  & {} = \mathit{merge}\left(M_{E\backslash F},M_{F}\right)
  \\
  &  {} = M_{E}
\end{align*}

We show next that when
$n\geq k\cdot N$, generating a set $E\in\mathcal{P}_{n}\left(\mathcal{U}\right)$
uniformly randomly can be seen as generating $k$ \emph{independent} sets in
$\mathcal{P}_{N}\left(\mathcal{U}\right)$, then merging them. Indeed, generating
such a set can be done by as follows:
\begin{enumerate}
\item For $i\in\left\{ 1,\ldots,k\right\} $, generate a set
$E_{i}\subseteq\mathcal{P}_{N}\left(\mathcal{U}\right)$ uniformly randomly. Let
$E_{\cup}=\bigcup_{i}E_{i}$.
\item Count the number of elements appearing in multiple $E_{i}$:
$d=\#\left\{ x\in E_{i}|\exists j<i,x\in E_{j}\right\} $. Generate a set
$E^{\prime}\in\mathcal{P}_{n-d}\left(\mathcal{U}\backslash E_{\cup}\right)$
uniformly randomly.
\end{enumerate}
$E$ is then defined by $E=E_{\cup}\cup E^{\prime}$. Step $1$ ensures
that we used $k$ \emph{independent} sets of cardinality $N$ to generate
$E$, and step $2$ ensures that $E$ has exactly $n$ elements.

Intuitively, each time we generate a set $E_i$ of cardinality $N$ uniformly at
random in $\mathcal{U}$, we have \emph{one chance} that $t$ will be ignored by
$E_i$ (and thus by $E$). So $t$ can be ignored by $M_{E}$ with a certain
probability because it was ignored by $M_{E_1}$. Similarly, it can also be
ignored because of $M_{E_2}$, etc. Since the choice of $E_i$ is independent of
the choice of elements in $\bigcup_{j\neq i}E_j$, we can rewrite:
\begin{align*}
 \mathbb{P}_{n}&\left[\mathit{add}\left(M_{E},t\right)\neq M_{E} \mid t\notin E\right]\\
 & {} \leq\prod_{i=1}^{k}\mathbb{P}_{n}\left[\mathit{add}\left(M_{E_{i}^{0}},t\right)\neq M_{E_{i}^{0}} \mid t\notin E\right]\\
 & {} \leq\prod_{i=1}^{k}\left(1-\mathbb{P}_{n}\left[\mathit{add}\left(M_{E_{i}^{0}},t\right)=M_{E_{i}^{0}} \mid t\notin E_{i}\right]\right)\\
 & {} \leq{\left(1-p\right)}^{k}
\end{align*}
using the hypothesis of the lemma. Thus:
\begin{equation*}
  \mathbb{P}_{n}\left[\mathit{add}\left(M_{E},t\right)=M_{E} \mid t\notin E\right]\geq1-{\left(1-p\right)}^{k}.
\end{equation*}
\end{proof}

\begin{lem}\label{lem:bad-variance}
Suppose a deterministic cardinality estimator satisfies
$\mathbb{P}_{n}\left[\mathit{add}\left(M_{E},t\right)=M_{E} \mid t\notin E\right]\geq1-p$
for any $n\geq N$ and all $t$. Then its variance for $n\geq N$ is at least
$\frac{1-p}{p}\left(n-N\right)$.
\end{lem}
\begin{proof}
The proof's intuition is as follows. The hypothesis of the lemma requires
that the cardinality estimator, on average, \emph{ignores} a proportion $1-p$ of
new elements added to a sketch (once $N$ elements have been added): the sketch
is not changed when a new element is added. The best thing that the cardinality
estimator can do, then, is to store all elements that it does not ignore, count
the number of unique elements among these, and multiply this number by $1/p$ to
correct for the elements ignored. It is well-known that estimating the size $k$
of a set based on the size of a uniform sample of sampling ratio $p$ has a
variance of $\frac{1-p}{p}k$. Hence, our cardinality estimator has a variance of
at least $\frac{1-p}{p}\left(n-N\right)$.

Formalizing this idea requires some additional technical steps. The full proof
is given in Appendix~\ref{lem:bad-variance-proof}.
\end{proof}

All existing cardinality estimators satisfy our axioms and their standard error
remains low even for large values of $n$. Theorem~\ref{thm:main-result} shows,
for all of them, that there are some users whose privacy loss is
significant. In Section~\ref{subsec:averaging-epsilon}, we quantify this
precisely for HyperLogLog.

\subsection{Probabilistic case\label{subsec:main-result-probabilistic}}

Algorithms that add noise to their output, or more generally, are allowed to use
a source of randomness, are often used in privacy contexts. As such, even though
all cardinality estimators used in practical applications are deterministic, it
is reasonable to hope that a probabilistic cardinality estimator could satisfy
our very weak privacy definition. Unfortunately, this is not the case.

In the deterministic case, we showed that for any element $t$, the probability
that $t$ has an influence on a random sketch $M$ decreases exponentially with
the sketch size. Or, equivalently, the distribution of sketches of size $kn$
that do not contain $t$ is ``almost the same'' (up to a density of probability
${\left(1-e^{-\varepsilon}\right)}^{k}$) as the distribution of sketches of the
same size, but containing $t$.

The following lemma establishes the same result in the probabilistic setting.
Instead of reasoning about the probability that an element $t$ is ``ignored'' by
a sketch $M$, we reason about the probability that $t$ has a \emph{meaningful}
influence on this sketch. We show that this probability decreases exponentially,
even if $\mathbb{P}\left[M\neq\mathit{add}\left(M,t\right)\right]$ is very high.

First, we prove a technical lemma on the \emph{structure} that the
$\mathit{merge}$ operation imposes on the space of sketch distributions. Then,
we find an upper bound on the ``meaningful influence'' of an element $t$, when
added to a random sketch of cardinality $n$. We then use this upper bound,
characterized using the statistical distance, to show that the estimator
variance is as imprecise as for the deterministic case.

\begin{defn}
Let $\mathbb{D}$ be the real vector space spanned by the family $\left\{ \mathcal{D}_{E}|E\subseteq\mathcal{U}\right\} $
(seen as vectors of $\mathbb{R}^{\mathcal{M}}$). For any probability
distributions $\mathcal{A},\mathcal{B}\in\mathbb{D}$, we denote $\mathcal{A}\cdot\mathcal{B}=\mathit{merge}\left(\mathcal{A},\mathcal{B}\right)$.
We show in Lemma~\ref{lem:merge-algebra} that this notation makes sense: on
$\mathbb{D}$, we can do computations as if $\mathit{merge}$ was a multiplicative
operation.
\end{defn}

\begin{lem}\label{lem:merge-algebra}
The $\mathit{merge}$ operation defines a commutative and associative
algebra on $\mathbb{D}$.

\begin{proof}
By the properties required from probabilistic cardinality estimators in
Definition~\ref{defn:probabilistic}, the $\mathit{merge}$ operation is commutative
and associative on the family
$\left\{\mathcal{D}_{E}|E\subseteq\mathcal{U}\right\}$. By linearity of the
$\mathit{merge}$ operation, these properties are preserved for any linear
combination of vectors $\mathcal{D}_{E}$.
\end{proof}
\end{lem}

\begin{lem}\label{lem:statistical-distance}
Suppose a cardinality estimator satisfies $\varepsilon$-sketch privacy above
cardinality $N$, and let $t\in\mathcal{U}$. Let $\mathcal{D}_{\text{out},n}$ be
the distribution of sketches obtained by adding $n$ uniformly random elements of
$\mathcal{U}\backslash\left\{ t\right\}$ into $M_{\varnothing}$ (or,
equivalently,
$\mathcal{D}_{\text{out},n}\left(M\right)=\mathbb{P}_{n}\left[M_E=M|t\notin E\right]$).
Then:
\[
\upsilon\left(\mathcal{D}_{\text{out},kn},\mathit{add}\left(\mathcal{D}_{\text{out},kn},t\right)\right)\leq{\left(1-e^{-\varepsilon}\right)}^{k}
\]
where $\upsilon$ is the statistical distance between probability distributions.
\end{lem}

\begin{proof}
Let $\mathcal{D}_{\text{in},n}$ be the distribution of sketches obtained
by adding $t$, then $n-1$ uniformly random elements of $\mathcal{U}$
into $M$ (or, equivalently, $\mathcal{D}_{\text{in},n}\left(M\right)=\mathbb{P}_{n}\left[M_E=M|t\in E\right]$).
Then the definition of $\varepsilon$-sketch privacy gives that for
every sketch $M$, $\mathcal{D}_{\text{out},n}\left(M\right)\geq e^{-\varepsilon}\mathcal{D}_{\text{in},n}\left(M\right)$.
So we can express $\mathcal{D}_{\text{out},n}$ as the \emph{sum}
of two distributions:
\[
\mathcal{D}_{\text{out},n}=e^{-\varepsilon}\mathcal{D}_{\text{in},n}+\left(1-e^{-\varepsilon}\right)\mathcal{R}
\]
for a certain distribution $\mathcal{R}$.

First, we show that
$\mathcal{D}_{\text{out},kn}={\left(\mathcal{D}_{\text{out},n}\right)}^{k}\cdot\mathcal{C}$
for a certain distribution $\mathcal{C}$. Indeed, to generate a sketch
of cardinality $kn$ that does not contain $t$ uniformly randomly, one can use
the following process.
\begin{enumerate}
\item Generate $k$ random sketches of cardinality $n$ which do not contain $t$,
and merge them.
\item For all $E\subseteq\mathcal{U}$, denote by $p_{E}$ the probability that
the $k$ sketches were generated with the elements in $E$. There might be
``collisions'' between the $k$ sketches: if several sketches were generated
using the same element, $\left|E\right|<kn$. When this happens, we need to
``correct'' the distribution, and add additional elements. Enumerating all the
options, we denote $\mathcal{C}=\sum p_{E}\mathcal{D}_{E,nk}^{\text{c}}$, where
$\mathcal{D}_{E,nk}^{\text{c}}$ is obtained by adding $nk-\left|E\right|$
uniformly random elements in $\mathcal{U}\backslash E$ to
$\mathcal{M}_{\varnothing}$. Thus,
$\mathcal{D}_{\text{out},kn}={\left(\mathcal{D}_{\text{out},n}\right)}^{k}\cdot\mathcal{C}$.
\end{enumerate}

All these distributions are in $\mathbb{D}$:
$\mathcal{D}_{\text{out},n}=\text{avg}_{E\in\mathcal{P}_{n}\left(\mathcal{U}\right),t\notin E}\mathcal{D}_{E}$,
$\mathcal{D}_{\text{in},n}=\text{avg}_{E\in\mathcal{P}_{n}\left(\mathcal{U}\right),t\in E}\mathcal{D}_{E}$,
$\mathcal{R}={\left(1-e^{-\varepsilon}\right)}^{-1}\left(\mathcal{D}_{\text{out},n}-e^{-\varepsilon}\mathcal{D}_{\text{in},n}\right)$,
etc. Thus:

\begin{align*}
  \mathcal{D}_{\text{out},kn} & ={\left(\mathcal{D}_{\text{out},n}\right)}^{k}\cdot\mathcal{C}\\
 & ={\left(e^{-\varepsilon}\mathcal{D}_{\text{in},n}+\left(1-e^{-\varepsilon}\right)\mathcal{R}\right)}^{k}\cdot\mathcal{C}\\
 & =\sum_{i=0}^{k}\binom{k}{i}e^{-i\cdot\varepsilon}{\left(1-e^{-\varepsilon}\right)}^{k-i}\mathcal{D}_{\text{in},n}^{i}\cdot\mathcal{R}^{k-i}\cdot\mathcal{C}.
\end{align*}
Denoting
$\mathcal{A}=\sum_{i=1}^{k}\binom{k}{i}e^{-i\cdot\varepsilon}{\left(1-e^{-\varepsilon}\right)}^{k-i}\mathcal{D}_{\text{in},n}^{i-1}\cdot\mathcal{R}^{k-i}\cdot\mathcal{C}$
and $\mathcal{\mathcal{B}=\mathcal{R}}^{k}\cdot\mathcal{C}$, this
gives us:
\[
\mathcal{D}_{\text{out},kn}=\mathcal{A}\cdot\mathcal{D}_{\text{in},n}+{\left(1-e^{-\varepsilon}\right)}^{k}\mathcal{B}.
\]

Finally, we can compute $\mathit{add}\left(\mathcal{D}_{\text{out},kn},t\right)$:
\begin{align*}
\mathit{add}\left(\mathcal{D}_{\text{out},kn},t\right) &
  =\mathcal{A}\cdot\mathcal{D}_{\text{in},n}\cdot\mathcal{D}_{\left\{ t\right\} }+{\left(1-e^{-\varepsilon}\right)}^{k}\mathcal{B}\cdot\mathcal{D}_{\left\{ t\right\} }\\
  & =\mathcal{A}\cdot\mathcal{D}_{\text{in},n}+{\left(1-e^{-\varepsilon}\right)}^{k}\mathcal{B}\cdot\mathcal{D}_{\left\{ t\right\} }
\end{align*}

Note that since
$\mathcal{D}_{\text{in},n}=\text{avg}_{E\in\mathcal{P}_{n}\left(\mathcal{U}\right),t\in E}\mathcal{D}_{E}$,
we have
$\mathcal{D}_{\text{in},n}\cdot\mathcal{D}_{\left\{t\right\}}=\mathcal{D}_{\text{in},n}$
by idempotence, and:
\begin{align}
  \nonumber
  \upsilon&\left(\mathcal{D}_{\text{out},kn},\mathit{add}\left(\mathcal{D}_{\text{out},kn},t\right)\right)  
  \\
  \nonumber
  & {} = \frac{1}{2}\left\Vert \mathcal{D}_{\text{out},kn}-\mathit{add}\left(\mathcal{D}_{\text{out},kn},t\right)\right\Vert _{1}
  \\
  \nonumber
  & {} = \frac{1}{2}\left\Vert {\left(1-e^{-\varepsilon}\right)}^{k}\mathcal{B}-{\left(1-e^{-\varepsilon}\right)}^{k}\mathcal{B}\cdot\mathcal{D}_{\left\{ t\right\} }\right\Vert _{1}
  \\
  \nonumber
  & {} \leq \frac{{\left(1-e^{-\varepsilon}\right)}^{k}}{2}\left(\left\Vert \mathcal{B}\right\Vert _{1}+\left\Vert \mathcal{B}\cdot\mathcal{D}_{\left\{ t\right\} }\right\Vert _{1}\right)
  \\
  \nonumber
  & {} \leq {\left(1-e^{-\varepsilon}\right)}^{k}.
\end{align}
\end{proof}

Lemma~\ref{lem:statistical-distance} is the probabilistic equivalent of
Lemmas~\ref{lem:ignores-elements} and~\ref{lem:ignores-more-and-more-elements}.
Now, we state the equivalent of Lemma~\ref{lem:bad-variance}, and explain why
its intuition still holds in the probabilistic case.

\begin{lem}\label{lem:bad-variance-probabilistic}
Suppose that a cardinality estimator satisfies for any $n\geq N$ and all $t$,
$\upsilon\left(\mathcal{D}_{\text{out},n},\mathit{add}\left(\mathcal{D}_{\text{out},n},t\right)\right)\leq p$.
Then its variance for $n\geq N$ is at least
$\frac{1-p}{p}\left(n-N\right)$.
\end{lem}

\begin{proof}
The condition
``$\upsilon\left(\mathcal{D}_{\text{out},n},\mathit{add}\left(\mathcal{D}_{\text{out},n},t\right)\right)\leq p$''
is equivalent to the condition of Lemma~\ref{lem:bad-variance}: with
probability $\left(1-p\right)$, the cardinality estimator ``ignores'' when a new
element $t$ is added to a sketch. Just like in Lemma~\ref{lem:bad-variance}'s
proof, we can convert this constraint into estimating the size of a set based on
a sampling set. The best known estimator for this problem is deterministic, so
allowing the cardinality estimator to be probabilistic does not help improving
the optimal variance.

The same result than in Lemma~\ref{lem:bad-variance} follows.
\end{proof}

Lemmas~\ref{lem:statistical-distance} and~\ref{lem:bad-variance-probabilistic}
together immediately lead to the equivalent of Theorem~\ref{thm:main-result} in
the probabilistic case.

\begin{thm}\label{thm:main-result-probabilistic}
An unbiased probabilistic cardinality estimator that satisfies
$\varepsilon$-sketch privacy above cardinality $N$ is not precise. Namely, its
variance is at least $\frac{1-c^{k}}{c^{k}}\left(n-k\cdot N\right)$, for any
$n\geq N$ and $k\leq\frac{n}{N}$, where $c = 1-e^{-\varepsilon}$
\end{thm}

Somewhat surprisingly, allowing the algorithm to add noise to the data seems to
be pointless from a privacy perspective. Indeed, given the same privacy
guarantee, the lower bound on accuracy is the same for deterministic and
probabilistic cardinality estimators. This suggests that the constraints of
these algorithms (idempotence and commutativity) require them to somehow keep a
trace of who was added to the sketch (at least for some users), which is
fundamentally incompatible with even weak notions of privacy.

\section{Weakening the privacy definition\label{sec:weaker-versions}}

Our main result is negative: no cardinality estimator satisfying our privacy
definition can maintain a good accuracy. Thus, it is natural to wonder whether
our privacy definition is too strict, and if the result still holds for weaker
variants.

In this section, we consider two weaker variants of our privacy definition: one
allows a small probability of privacy loss, while the other averages the privacy
loss across all possible outputs. We show that these natural relaxations do not
help as close variants of our negative result still hold.

\subsection{Allowing a small probability of privacy loss}

\noindent
As Lemma~\ref{lem:equivalent-definition} shows, $\varepsilon$-sketch
differential privacy provides a bound on how much information the attacker can
gain in the worst case. A natural relaxation is to accept a small
\emph{probability of failure}: requiring a bound on the information gain in
\emph{most cases}, and accept a potentially unbounded information gain with low
probability.

We introduce a new parameter, called $\delta$, similar to the use of $\delta$ in
the definition of $\left(\varepsilon,\delta\right)$-differential privacy:
$\mathcal{A}$ is $\left(\varepsilon,\delta\right)$-differentially private if and
only if for any databases $D_1$ and $D_2$ that only differ by one element and
any set $S$ of possible outputs,
$\mathbb{P}\left[\mathcal{A}\left(D_{1}\right)\in S\right]
   \leq e^{\varepsilon}\cdot\mathbb{P}\left[\mathcal{A}\left(D_{2}\right)\in S\right]+\delta$.

\begin{defn}\label{def:epsilon-delta-sketch-privacy}
A cardinality estimator satisfies \emph{$\left(\varepsilon,\delta\right)$-sketch
privacy above cardinality $N$} if for every $\mathcal{S\subseteq M}$, $n\geq N$,
and $t\in\mathcal{U}$,
\[
\mathbb{P}_{n}\left[M_{E}\in\mathcal{S} \mid t\in E\right]
  \leq e^{\varepsilon}\cdot\mathbb{P}_{n}\left[M_{E}\in\mathcal{S} \mid t\notin E\right]+\delta.
\]
\end{defn}

Unfortunately, our negative result still holds for this variant of the
definition. Indeed, we show that a close variant of
Lemma~\ref{lem:ignores-elements} holds, and the rest follows directly.

\begin{lem}\label{lem:ignores-elements-with-delta}
Let $t\in\mathcal{U}$. A cardinality estimator that satisfies
$\left(\varepsilon,\delta\right)$-probabilistic sketch privacy above cardinality
$N$ satisfies
$\mathbb{P}_{n}\left[\mathit{add}\left(M_{E},t\right)=M_{E} \mid t\notin E\right]
    \geq\left(\frac{1}{2}-\delta\right)\cdot e^{-\varepsilon}$
for $n\geq N$.
\end{lem}

The proof of Lemma~\ref{lem:ignores-elements-with-delta} is given in
Appendix~\ref{lem:ignores-elements-with-delta-proof}. We can then deduce a
theorem similar to our negative result for our weaker privacy definition.

\begin{thm}\label{thm:negative-result-with-delta}
An unbiased cardinality estimator that satisfies $\left(\varepsilon,\delta\right)$-sketch
privacy above cardinality $N$ has a variance at least
$\frac{1-c^{k}}{c^{k}}\left(n-k\cdot N\right)$ for any $n\leq N$ and
$k\leq\frac{n}{N}$, where $c = 1-\left(\frac{1}{2}-\delta\right)\cdot e^{-\varepsilon}$.
It is therefore not precise if $\delta < \frac{1}{2}$.
\end{thm}
\begin{proof}
This follows from
Lemmas~\ref{lem:ignores-elements-with-delta},~\ref{lem:ignores-more-and-more-elements}
and~\ref{lem:bad-variance}.
\end{proof}

\subsection{Averaging the privacy loss}

\noindent
Instead of requiring that the attacker's information gain is bounded by
$\varepsilon$ for every possible output, we could bound the \emph{average}
information gain. This is equivalent to accepting a larger privacy loss in some
cases, as long as other cases have a lower privacy loss.

This intuition is captured by the use of Kullback-Leiber divergence, which is
often used in similar
contexts~\cite{rebollo2010optimized,rebollo2010t,diaz2002towards,dwork2010boosting}.
In our case, we adapt it to maintain the asymmetry of our original privacy
definition. First, we give a formal definition the \emph{privacy loss} of a user
$t$ given output $M$.

\begin{defn}\label{def:privacy-loss}
Given a cardinality estimator, the \emph{positive privacy loss of $t$ given
output $M$ at cardinality $n$} is defined as
\[
\varepsilon_{n,t,M}=\max\left(\log\left(\frac{\mathbb{P}_{n}\left[M_E=M \mid t\in
E\right]}{\mathbb{P}_{n}\left[M_E=M \mid t\notin E\right]}\right),0\right).
\]
\end{defn}

This privacy loss is never negative: this is equivalent to discarding the case
where the attacker gains \emph{negative} information. Now, we bound this average
over all possible values of $M_E$, given $t\in E$.

\begin{defn}
A cardinality estimator satisfies \emph{$\varepsilon$-sketch average
privacy above cardinality $N$} if for every $n\geq N$ and $t\in\mathcal{U}$,
we have
\[
\sum_{M}\mathbb{P}_{n}\left[M_E=M \mid t\in E\right]\cdot\varepsilon_{n,t,M}\leq\varepsilon.
\]
\end{defn}

It is easy to check that $\varepsilon$-sketch average privacy above cardinality
$N$ is strictly weaker than $\varepsilon$-sketch privacy above cardinality $N$.
Unfortunately, this definition is also stronger than
$\left(\varepsilon_\delta,\delta\right)$-sketch privacy above cardinality $N$
for certain values of $\varepsilon$ and $\delta$, and as such,
Lemma~\ref{lem:ignores-elements-with-delta} also applies. We prove this in the
following lemma.

\begin{lem}\label{lem:average-implies-delta}
If a cardinality estimator satisfies $\varepsilon$-sketch average privacy above
cardinality $N$, then it also satisfies
$\left(\frac{\varepsilon}{\delta},\delta\right)$-sketch privacy above
cardinality $N$ for any $\delta>0$.
\end{lem}

The proof is given in Appendix~\ref{lem:average-implies-delta-proof}. This lemma
leads to a similar version of the negative result.

\begin{thm}
An unbiased cardinality estimator that satisfies $\varepsilon$-sketch
average privacy above cardinality $N$ has a variance at least
$\frac{1-c^{k}}{c^{k}}\left(n-k\cdot N\right)$ for any $n\leq N$ and
$k\leq\frac{n}{N}$, where $c = 1-\frac{e^{-4\varepsilon}}{4}$.
It is thus not precise.
\end{thm}
\begin{proof}
This follows directly from Lemma~\ref{lem:average-implies-delta} with
$\delta=\frac{1}{4}$, and Theorem~\ref{thm:negative-result-with-delta}.
\end{proof}

Recall that all existing cardinality estimators satisfy our axioms and have a
bounded accuracy. Thus, an immediate corollary is that for all cardinality
estimators used in practice, there are some users for which the \emph{average}
privacy loss is very large.

\begin{rem}\label{rem:renyi-privacy}
This idea of \emph{averaging $\varepsilon$} is similar to the idea behind Rényi
differential privacy~\cite{mironov2017renyi}. The parameter $\alpha$ of Rényi
differential privacy determines the averaging method used (geometric mean,
arithmetic mean, quadratic mean, etc.). Using KL-divergence corresponds to
$\alpha=1$, while $\alpha=2$ averages all possible values of $e^\varepsilon$.
Increasing $\alpha$ strengthens the privacy
definition~\cite[Prop.~9]{mironov2017renyi}, so our negative result still
holds.
\end{rem}

\section{Privacy loss of individual users\label{sec:individual-users}}

So far, we only considered definitions of privacy that give the same guarantees
for all users. What if we allow certain users to have less privacy than others,
or if we were to average the privacy loss \emph{across users} instead of
averaging over all possible outcomes for each user?

Such definitions would generally not be sufficiently convincing to be used in
practice: one typically wants to protect \emph{all users}, not just a majority
of them. In this section, we show that even if we relax this requirement,
cardinality estimators would in practice leak a significant amount of
information.

\subsection{Allowing unbounded privacy loss for some users}

\noindent
What happens if we allow some users to have unbounded privacy loss? We could
achieve this by requiring the existence of a subset of users
$\mathcal{T}\subseteq\mathcal{U}$ of density $1-\delta$, such that every user
in $\mathcal{T}$ is protected by $\varepsilon$-sketch privacy above cardinality
$N$. In this case, a ratio $\delta$ of possible targets are not protected.

This approach only makes sense if the attacker cannot choose the target $t$. For
our attacker model, this might be realistic: suppose that the attacker wants to
target just one particular person. Since all user identifiers are hashed before
being passed to the cardinality estimator, this person will be associated to a
hash value that the attacker can neither predict nor influence. Thus, although
the attacker picks $t$, the true target of the attack is $h(t)$, which the
attacker cannot choose.

Unfortunately, this drastic increase in privacy risk for some users does not
lead to a large increase in accuracy. Indeed, the best possible use of this
ratio $\delta$ of users from an accuracy perspective would simply be to count
exactly the users in a sample of sampling ratio $\delta$.

Estimating the total cardinality based on this sample, similarly to the optimal
estimator in the proof of Lemma~\ref{lem:bad-variance}, leads to a variance of
$\frac{1-\delta}{\delta}\cdot\left(n-N\right)$. If $\delta$ is very small (say,
$\delta\simeq10^{-4}$), this variance is too large for counting small values of
$n$ (say, $n\simeq1000$ and $N\simeq100$). This is not surprising: if $99.99\%$
of the values are ignored by the cardinality estimator, we cannot expect it to
count values of $n$ on the order of thousands. But even this value of $\delta$
is larger than what is often used with
$\left(\varepsilon,\delta\right)$-differential privacy, where typically,
$\delta=o(1/n)$.

But in our running example, sketches must yield a reasonable accuracy both at
small and large cardinalities, if many sketches are aggregated. This implicitly
assumes that the service operates at a large scale, say with at least $10^{7}$
users. With $\delta=10^{-4}$, this means that thousands of users are not covered
by the privacy property. This is unacceptable for most applications.

\subsection{Averaging the privacy loss across users\label{subsec:averaging-epsilon}}

\noindent
Instead of requiring the same $\varepsilon$ for every user, we could require
that the \emph{average} information gain by the attacker is bounded by
$\varepsilon$. In this section, we take the example of HyperLogLog to show that
accuracy is not incompatible with this notion of average privacy, but that
cardinality estimators used in practice do not preserve privacy even if we
average across all users.

First, we define this notion of average information gain across users.

\begin{defn}
Recall the definition of the positive privacy loss $\varepsilon_{n,t,M}$ of $t$ given output $M$ at
cardinality $n$ from Definition~\ref{def:privacy-loss}:
The maximum privacy loss of $t$ at cardinalty $n$ is defined as
$\varepsilon_{n,t} = \max_{M}\left(\varepsilon_{n,t,M}\right)$.
A cardinality estimator satisfies \emph{$\varepsilon$-sketch privacy on average}
if we have, for all $n$,
$\frac{1}{\left|\mathcal{U}\right|}\sum_{t\in\mathcal{U}}\varepsilon_{n,t}\leq\varepsilon$.
\end{defn}

In this definition, we accept that some users might have less privacy as long as
the \emph{average} user satisfies our initial privacy definition.
Remark~\ref{rem:renyi-privacy} is still relevant: we chose to average over all
values of $\varepsilon$, but other averaging functions are possible and would
lead to strictly stronger definitions.

We show that HyperLogLog satisfies this definition and we consider the value of
$\varepsilon$ for various parameters and their significance. Intuitively, a
HyperLogLog cardinality estimator puts every element in a random \emph{bucket},
and each bucket counts the \emph{maximum number of leading zeroes} of elements
added in this bucket. More details are given in
Appendix~\ref{thm:hll:average-proof}.

HyperLogLog cardinality estimators have a parameter $p$ that determines its
memory consumption, its accuracy, and, as we will see, its level of average
privacy.

\begin{thm}\label{thm:hll:average}
Assuming a sufficiently large $\left|\mathcal{U}\right|$, a HyperLogLog
cardinality estimator of parameter $p$ satisfies $\varepsilon_{n}$-sketch
privacy above cardinality $N$ on average where for $N\ge n$,
\[
\varepsilon_{n} \simeq -\sum_{k\geq1}2^{-k}\log\left(1-{\left(1-2^{-p-k}\right)}^{n}\right).
\]
\end{thm}

The assumption that the set of possible elements is very large and its
consequences are explained in more detail in the proof of this theorem,
given in Appendix~\ref{thm:hll:average-proof}.

How does this positive result fit practical use cases?
Figure~\ref{fig:averagedepsilon} plots $\varepsilon_{n}$ for three different
HyperLogLog cardinality estimators. It shows two important results.

\begin{figure}
\begin{centering}
\includegraphics[height=4cm]{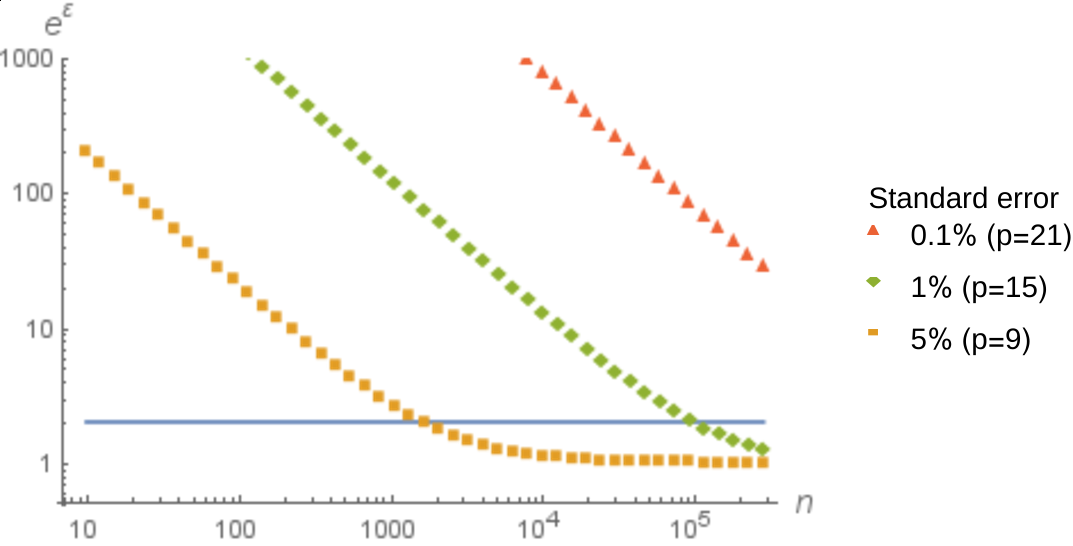}
\par\end{centering}
\caption{\label{fig:averagedepsilon}$\varepsilon_{n}$ as a function of $n$,
for HyperLogLog cardinality estimators of different $p$ parameters.
The blue line is $2$, corresponding to the commonly recommended value
of $\varepsilon=\ln\left(2\right)$ in differential privacy.}
\end{figure}

First, cardinality estimators used in practice do not preserve privacy. For
example, the default parameter used for production pipelines at Google and on
the BigQuery service~\cite{bigqueryhll} is $p=15$. For this value of $p$, an
attacker can determine with significant accuracy whether a target was added to a
sketch; not only in the worst case, but for the \emph{average }user too. The
average risk only becomes reasonable for $n\geq10,000$, a threshold too large
for most data analysis tasks.

Second, by sacrificing some accuracy, it is possible to obtain a reasonable
\emph{average privacy}. For example, a HyperLogLog sketch for which $p=9$ has a
relative standard error of about $5\%$, and an $\varepsilon_{n}$ of about $1$
for $n=1000$. Unfortunately, even when the average risk is acceptable, some
users will still be at a higher risk: users $e$ with a large number of leading
zeroes are much more identifiable than the average. For example, if $n=1000$,
there is a $98\%$ chance that at least one user has $\rho(e)\geq8$. For this
user, $\varepsilon_{n,t}\simeq5$, a very high value.

Our calculations yield only an approximation of $\varepsilon_{n}$ that is an
upper bound on the actual privacy loss in HyperLogLog sketches. However, these
alarming results can be confirmed experimentally. We simulated
$\mathbb{P}_{n}\left[\mathit{add}\left(M_{E},t\right)=M_{E} \mid t\notin
E\right]$, for uniformly random values of $t$, using HyperLogLog sketches with
the parameter $p=15$, the default used for production pipelines at Google and on
the BigQuery service~\cite{bigqueryhll}. For each cardinality $n$, we generated
10,000 different random target values, and added each one to 1,000 HyperLogLog
sketches of cardinality $n$ (generated from random values). For each target, we
counted the number of sketches that ignored it.

Figure~\ref{fig:ignoredelements} plots some percentile values. For example, the
all-targets curve (100th percentile) has a value of 33\% at cardinality $n$ =
10,000. This means that each of the 10,000 random targets was ignored by at most
33\% of the 1,000 random sketches of this cardinality, i.e.,
$\mathbb{P}_{n}\left[\mathit{add}\left(M_{E},t\right)=M_{E} \mid t\notin
E\right] \leq 33\%$ for all $t$. In other words, an attacker observes with at
least 67\% probability a change when adding a random target to a random sketch
that did not contain it. Similarly, the 10th-percentile at $n$ = 10,000 has a
value of 3.8\%. So 10\% of the targets were ignored by at most 3.8\% of the
sketches, i.e., $\mathbb{P}_{n}\left[\mathit{add}\left(M_{E},t\right)=M_{E} \mid
t\notin E\right] \leq 3.8\%$ for 10\% of all users $t$. That is, for the average
user $t$, there is a 10\% chance that a sketch with 10,000 elements changes with
likelihood at least 96.2\% when $t$ is first added.

\begin{figure}
\begin{centering}
\includegraphics[height=5cm]{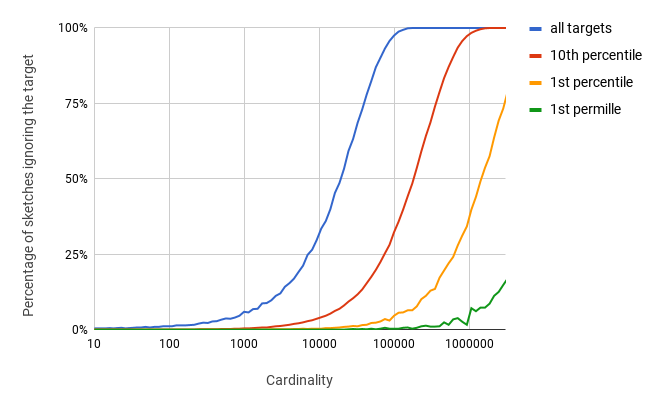}
\par\end{centering}
\caption{\label{fig:ignoredelements} Simulation of
$\mathbb{P}_{n}\left[\mathit{add}\left(M_{E},t\right)=M_{E} \mid t\notin E\right]$,
for uniformly chosen values of $t$, using HyperLogLog sketches with parameter
$p=15$.}
\end{figure}

For small cardinalities ($n<1,000$), adding an element that has not yet been
added to the sketch will almost certainly modify the sketch: an attacker
observing that a sketch does not change after adding $t$ can deduce with
near-certainty that $t$ was added previously.

Even for larger cardinalities, there is always a constant number of people with
high privacy loss. For $n$ = 1,000, no target was ignored by more than
5.5\% of the sketches. For $n$ = 10,000, 10\% of the users were ignored by at
most 3.8\% of the sketches. Similarly, the 1st percentile at $n$ = 100,000 and
the 1st permille at $n$ = 1,000,000 are 4.6\% and 4.5\%, respectively. In
summary, across all cardinalities $n$, at least 1,000 users $t$ have
$\mathbb{P}_{n}\left[\mathit{add}\left(M_{E},t\right)=M_{E} \mid t\notin
E\right] \leq 0.05$. For these users, the corresponding privacy loss is
$e^\varepsilon=\frac{1}{0.055} \simeq 18$. Concretely, if the attacker initially
believes that $\mathbb{P}_{n}\left[t\in E\right]$ is 1\%, this number grows to
15\% after observing that $\mathit{add}\left(M_{E},t\right)=M_{E}$. If it is
initially 10\%, it grows to 66\%. And if it is initially 25\%, it grows to 86\%.

\section{Mitigation strategies\label{sec:discussion}}

\noindent
A corollary of Theorem~\ref{thm:main-result} and of our analysis of
Section~\ref{subsec:averaging-epsilon} is that the cardinality estimators used
in practice do not preserve privacy. How can we best protect cardinality
estimator sketches against insider threats,in realistic settings? Of course,
classical data protection techniques are relevant: encryption, access controls,
auditing of manual accesses, etc. But in addition to these best practices,
cardinality estimators like HyperLogLog allow for specific risk mitigation
techniques, which restrict the attacker's capabilities.

\subsection{Salting the hash function with a secret\label{subsec:salting-the-hash}}

\noindent
As explained in Section~\ref{subsec:def-cardinality-estimators}, most
cardinality estimators use a hash function $h$ as the first step of the
$\mathit{add}$ operation: $\mathit{add}\left(M,t\right)$ only depends on $M$ and
the hash value $h(t)$. This hash can be salted with a secret value. This salt can
be made inaccessible to humans, with access controls restricting access to
production binaries compiled from trusted code. Thus, an adversary cannot learn
all the relevant parameters of the cardinality estimator and can no longer add
users to sketches. Of course, to avoid a salt reconstruction attack, a
cryptographic hash function must be used.

The use of a salt does not hinder the usefulness of sketches: they can still be
merged (for all cardinality estimators given as examples in
Section~\ref{subsec:def-cardinality-estimators}) and the cardinality can still
be estimated without accuracy loss. However, if an attacker gains direct access
to a sketch $M$ with the aim of targeting a user $t$ and does not know the
secret salt, then she cannot compute $h\left(t\right)$ and therefore  cannot
compute $\mathit{add}\left(M,t\right)$. This prevents the previous obvious
attack of adding $t$ to $M$ and observing whether the result is different.

However, this solution has two issues.

\begin{description}
\item[Secret salt rotation] The secret salt must be the same for all sketches as
otherwise sketches cannot be merged.   Indeed, if a hash function $h_{1}$ is
used to create a sketch $M_{1}$ and $h_{2}$ is used to create $M_{2}$, then if
$h_{1}\left(t\right)\neq h_{2}\left(t\right)$ for some $t$ that is added to both
$M_{1}$ and $M_{2}$, $t$ will be seen as a \emph{different user} in $M_{1}$ and
$M_{2}$: the cardinality estimator no longer ignores duplicates. Good key
management practices also recommend regularly rotating secret keys. In this
context, changing the key requires recomputing all previously computed sketches.
This requires keeping the original raw data, makes pipelines more complex, and
can be computationally costly.
\item[Sketch intersection] For most cardinality estimators given as examples in
Section~\ref{subsec:def-cardinality-estimators}, it is possible for an attacker
to guess $h\left(t\right)$ from a family of sketches ($M_{1}$, \ldots, $M_{k}$)
for which the attacker \emph{knows} that $t\in M_{1}$. For example, intersecting
the lists stored in K-Minimum Values sketches can provide information on which
hashes come from users that have been in \emph{all} sketches. For HyperLogLog,
one can use the leading zeroes in non-empty buckets to get partial information
on the hash value of users who are in all sketches. Moreover,
HyperLogLog++~\cite{heule2013hyperloglog} has a \emph{sparse mode} that stores
full hashes when the sketch contains a small number of values; this makes
intersection attacks even easier.

Intersection attacks are realistic, although they are significantly more complex
than simply checking if $\mathit{add}\left(M,t\right)=M$. In our running
example, sketches come from counting users across locations and time periods. If
an internal attacker wants to target someone she knows, she can gather
information about where they went using side channels like social media posts.
This gives her a series of sketches $M_{1}, \ldots, M_{k}$ that she \emph{knows}
her target belongs to, and from these, she can get information on $h(t)$ and use
it to perform an attack equivalent to checking whether
$\mathit{add}\left(M,t\right)=M$.
\end{description}

Another possible risk mitigation technique is homomorphic encryption. Each
sketch could be encrypted in a way that allows sketches to be merged, and new
elements to be added; while ensuring that an attacker cannot do any operation
without some secret key. Homomorphic encryption typically has significant
overhead, so it is likely to be too costly for most use cases. Our impossibility
results assume a computationally unbounded attacker; however, it is possible
that an accurate sketching mechanism using homomorphic encryption could provide
privacy against \emph{polynomial-time} attackers. We leave this area of research
for future work.

\subsection{Using a restricted API\label{subsec:restricted-api}}

\noindent
Using cardinality estimator sketches to perform data analysis tasks only
requires access to two operations: $\mathit{merge}$ and $\mathit{estimate}$. So
a simple option is to process the sketches over an API that only allows this
type of operation. One option is to provide a SQL engine on a database, and only
allow SQL functions that correspond to $\mathit{merge}$ and $\mathit{estimate}$
over the column containing sketches. In the BigQuery SQL engine, this
corresponds to allowing HLL\_COUNT.MERGE and HLL\_COUNT.EXTRACT functions, but
not other functions over the column containing sketches~\cite{bigqueryhll}.
Thus, the attacker cannot access the raw sketches.

Under this technique, an attacker who only has access to the API can no longer
directly check whether $\mathit{add}\left(M,t\right)=M$. Since she does not have
access to the sketch internals, she cannot perform the intersection attack
described previously either. To perform the check, her easiest option is to
impersonate her target within the service, interact with the service so that a
sketch $M_{\{t\}}$ containing \emph{only} her target is created in the sketch
database, and compare the estimates obtained from $M$ and
$\mathit{merge}\left(M,M_{\{t\}}\right)$.
Following the intuition given in Section~\ref{subsec:main-result-deterministic},
if these estimates are the same, then the target is more likely to be in the
dataset. How much information the attacker gets this way depends on
$\mathbb{P}\left[\mathit{estimate}\left(\mathit{add}\left(M_{E},t\right)\right)
    = \mathit{estimate}\left(M_{E}\right) \mid t\notin E\right]$.
We can increase this quantity by rounding the result of the $\mathit{estimate}$
operation, thus limiting the accuracy of the external attacker. This would make
the attack described in this work slightly more difficult to execute, and less
efficient. However, it is likely that the attack could be adapted, for example
by repeating it multiple times with additional fake elements.

This risk mitigation technique can be combined with the previous one. The
restricted API protects the sketches during normal use by data analysts, i.e.,
against external attackers. The hash salting mitigates the risk of manual access
to the sketches, e.g., by internal attackers. This type of direct access is not
needed for most data analysis tasks, so it can be monitored via other means.

\section{Conclusion}

We formally defined a class of cardinality estimator algorithms with an
associated system and attacker model that captures the risks associated with
processing personal data in cardinality estimator sketches. Based on this model,
we proposed a privacy definition that expresses that the attacker cannot gain
significant knowledge about a given target.

We showed that our privacy definition, which is strictly weaker than any
reasonable definition used in practice, is incompatible with the accuracy and
aggregation properties required for practical uses of cardinality estimators. We
proved similar results for even weaker definitions, and we measured the privacy
loss associated with the HyperLogLog cardinality estimator, commonly used in
data analysis tasks.

Our results show that designing accurate privacy-preserving cardinality
estimator algorithms is impossible, and that the cardinality estimator sketches
used in practice should be considered as sensitive as raw data. These negative
results are a consequence of the \emph{structure} imposed on cardinality
estimators: idempotence, commutativity, and existence of a well-behaved merge
operation. This result shows a fundamental incompatibility between accurate
aggregation and privacy. A natural question is ask whether other sketching
algorithms have similar incompatibilities, and what are minimal axiomatizations
that lead to similar impossibility results.

\section{Acknowledgements}

The authors thank Jakob Dambon for providing insights which helped us prove
Lemma~\ref{lem:bad-variance}, Esfandiar Mohammadi for numerous fruitful
discussions, as well as Christophe De Cannière, Pern Hui Chia, Chao Li, Aaron
Johnson and the anonymous reviewers for their helpful comments. This work was
partially funded by Google, and done while Andreas Lochbihler was at ETH Zurich.

\bibliographystyle{ACM-Reference-Format}
\bibliography{biblio}

\appendix

\section{Proof of Lemma~\ref{lem:merge-properties}\label{lem:merge-properties-proof}}

Let $M$ and $M^\prime$ be two sketches. If $E=\left\{ e_{1},\ldots,e_{k}\right\} $,
we denote by $\mathit{add}\left(M,E\right)$ the result of adding elements of $E$
successively to $M$:
$\mathit{add}\left(M,E\right)=\mathit{add}\left(\dots\mathit{add}\left(\mathit{add}\left(M,e_{1}\right),e_{2}\right)\dots,e_{k}\right)$.

Let $E_{1}$ and $E_{2}$ be two sets such that $M_{E_{1}}=M_{E_{2}}=M^{\prime}$,
and let $E$ be a set such as $M_{E}=M$. Then
$\mathit{add}\left(M,E_{1}\right)=\mathit{add}\left(\mathit{add}\left(M_{\varnothing},E\right),E_{1}\right)$.
Using the two properties of the add function, we get
$\mathit{add}\left(M,E_{1}\right)=\mathit{add}\left(\mathit{add}\left(M_{\varnothing},E_{1}\right),E\right)=\mathit{add}\left(M^{\prime},E\right)$.
The same reasoning leads to
$\mathit{add}\left(M,E_{2}\right)=\mathit{add}\left(M,E_{1}\right)$: the merge
function does not depend on the choice of $E$.

In addition, note that
$\mathit{merge}\left(M,M^{\prime}\right)=\mathit{add}\left(M_{\varnothing},E\cup E_{1}\right)$.
Thus, commutativity, associativity, idempotence, and neutrality follow directly
from the same properties of set union.

\section{Full proof of Lemma~\ref{lem:bad-variance}\label{lem:bad-variance-proof}}

The proof is decomposed into three steps. In Step 1, we fix the first $N$
elements added to the sketch, and we bound the variance of the estimator that
has these elements as initial input. In Step 2, we explicitly compute the
probability for an element $t$ to be ignored by $M_E$, first when $t$ is fixed,
then when $E$ is fixed. We then use these results in Step 3 to average the bound
over all possible choices for the $N$ elements in $E$, which gives us an overall
bound.

\vspace{3mm}\noindent\emph{Step 1: Bound the estimator variance.}

Let $E\in\mathcal{P}_{N}\left(\mathcal{U}\right)$. Let $S_{E}=\{ t\in\mathcal{U} \mid
\mathit{add}\left(M_{E},t\right)\neq M_{E}\}$ be the set of elements that
are \emph{not} ignored by $M_{E}$. Let
$p_{E}=\frac{\left|S_{E}\right|}{\left|\mathcal{U}\right|}$, or equivalently,
let $p_{E}=\mathbb{P}_{t}\left[\mathit{add}\left(M_{E},t\right)\neq M_{E} \mid
t\notin E\right]$, where $\mathbb{P}_{t}$ is the distribution that picks $t$
uniformly randomly in $\mathcal{U}$.

$S_{E}$ can be seen as the \emph{sampling set} of the estimator with
$E$ as initial input, while $p_{E}$ can be seen as its \emph{sampling
fraction}: all other elements are discarded, so the estimator only
has access to elements in $E$ and $S_{E}$ to compute its estimate.

The optimal estimator to minimize variance in this context simply counts the
exact number of distinct elements in the sample (remembering each one), and
divides this number by $p_{E}$ to estimate the total number of distinct
elements.\footnote{The optimality of an estimator under these constraints is
proven e.g.\ in~\cite{casella2002statistical}.}

What is the variance $\mathbb{V}_{n|E}$ of this optimal estimator?
Suppose we added $n-N$ elements after reaching the sampling part.
The number of elements in the sample is a random variable with variance
$p_{E}\left(1-p_{E}\right)\left(n-N\right)$. Dividing this random variable by
$p_{E}$ gives a variance of $\frac{1-p_{E}}{p_{E}}\cdot\left(n-N\right)$.
Thus
\[
\mathbb{V}_{n|E}\geq\frac{1-p_{E}}{p_{E}}\cdot\left(n-N\right).
\]

Since the first $N$ elements are chosen uniformly at random, the variance
of the overall estimator is bounded by their average. If we denote by
$\mathcal{P}_{N}\left(\mathcal{U}\right)$ the set of all possible subsets
$E\subseteq\mathcal{U}$ of cardinality $N$, we have:
\begin{equation}
\mathbb{V}_{n}\geq\text{avg}_{E\in\mathcal{P}_{N}\left(\mathcal{U}\right)}\frac{1-p_{E}}{p_{E}}\left(n-N\right)\label{eq:average-variance}
\end{equation}
(where $\text{avg}$ stands for the average).

\vspace{3mm}\noindent\emph{Step 2: Intermediary results.}

Fix $E\in\mathcal{P}_{N}\left(\mathcal{U}\right)$. Denoting by $1_X$ the
function whose value is 1 if $X$ is satisfied and 0 otherwise,
\begin{align}
  \nonumber
  \mathbb{P}_{t}&\left[\mathit{add}\left(M_{E},t\right)\neq M_{E} \mid t\notin E\right] 
  \\
  \nonumber
  & {} = \frac{\mathbb{P}_{t}\left[\mathit{add}\left(M_{E},t\right)\neq M_{E}\wedge t\notin E\right]}{\mathbb{P}_{t}\left[t\notin E\right]} 
  \\
  & {} = \frac{\left|\mathcal{U}\right|}{\left|\mathcal{U}\right|-N}\cdot\text{avg}_{t\in\mathcal{U}}\left(1_{\mathit{add}\left(M_{E},t\right)\neq M_{E}}\right).
  \label{eq:E-fixed}
\end{align}
Indeed, $\mathbb{P}_{t}\left[t\notin E\right]=\frac{\left|\mathcal{U}\right|-N}{\left|\mathcal{U}\right|}$
is straightforward, and since $t\in E$ implies $\mathit{add}\left(M_{E},t\right)=M_{E}$,
the condition $\mathit{add}\left(M_{E},t\right)\neq M_{E}\wedge t\notin E$
can be simplified to $\mathit{add}\left(M_{E},t\right)\neq M_{E}$.

Now, fix $t\in\mathcal{U}$. Then
\begin{align}
  \nonumber
  \mathbb{P}_{N}&\left[\mathit{add}\left(M_{E},t\right)\neq M_{E} \mid t\notin E\right] 
  \\
  \nonumber
  & {} = \frac{\mathbb{P}_{N}\left[\mathit{add}\left(M_{E},t\right)\neq M_{E}\wedge t\notin E\right]}{\mathbb{P}_{N}\left[t\notin E\right]}
  \\
  & {} = \frac{\left|\mathcal{U}\right|}{\left|\mathcal{U}\right|-N}\cdot\text{avg}_{E\in\mathcal{P}_{N}\left(E\right)}\left(1_{\mathit{add}\left(M_{E},t\right)\neq M_{E}}\right).
  \label{eq:t-fixed}
\end{align}

Indeed, we have
\begin{align*}
  \mathbb{P}_{N}\left[t\notin E\right] 
  &
  {} = \binom{\left|\mathcal{U}\right|-1}{N}\cdot\binom{\left|\mathcal{U}\right|}{N}^{-1}
  \\
  & {} = \frac{\left(\left|\mathcal{U}\right|-1\right)!}{\left(\left|\mathcal{U}\right|-1-N\right)!N!}\cdot\frac{\left(\left|\mathcal{U}\right|-N\right)!N!}{\left|\mathcal{U}\right|!}
  = \frac{\left|\mathcal{U}\right|-N}{\left|\mathcal{U}\right|}
\end{align*}

\vspace{3mm}\noindent\emph{Step 3: Conclude using convexity.}

We now prove that
$\text{avg}_{E\in\mathcal{P}_{N}\left(\mathcal{U}\right)}\left(p_{E}\right)\leq p$.
Our initial hypothesis states that $p\geq\mathbb{P}_{N}\left[\mathit{add}\left(M_{E},t\right)\neq M_{E} \mid t\notin E\right]$ for all $t$.
We can average this for every $t$ and use~\eqref{eq:t-fixed} and~\eqref{eq:E-fixed}:
\begin{eqnarray}
p & \geq &
\text{avg}_{t\in\mathcal{U}}\left(\mathbb{P}_{N}\left[\mathit{add}\left(M_{E},t\right)\neq
M_{E} \mid t\notin E\right]\right)\nonumber \\
 & \overset{\eqref{eq:t-fixed}}{=} & \text{avg}_{t\in\mathcal{U}}\left(\frac{\left|\mathcal{U}\right|}{\left|\mathcal{U}\right|-N}\cdot\text{avg}_{E\in\mathcal{P}_{N}\left(E\right)}\left(1_{\mathit{add}\left(M_{E},t\right)\neq M_{E}}\right)\right)\nonumber \\
 & = & \text{avg}_{E\in\mathcal{P}_{N}\left(\mathcal{U}\right)}\left(\frac{\left|\mathcal{U}\right|}{\left|\mathcal{U}\right|-N}\cdot\text{avg}_{t\in\mathcal{U}}\left(1_{\mathit{add}\left(M_{E},t\right)\neq M_{E}}\right)\right)\nonumber \\
 & \overset{\eqref{eq:E-fixed}}{=} & \text{avg}_{E\in\mathcal{P}_{N}\left(\mathcal{U}\right)}\left(\mathbb{P}_{t}\left[\mathit{add}\left(M_{E},t\right)\neq M_{E} \mid t\notin E\right]\right)\nonumber \\
 & \geq & \text{avg}_{E\in\mathcal{P}_{N}\left(\mathcal{U}\right)}\left(p_{E}\right)\label{eq:averaging-p}
\end{eqnarray}

Now, using~\eqref{eq:average-variance} and~\eqref{eq:averaging-p} along
with the fact that the function
$x\longrightarrow\frac{1-x}{x}\cdot\left(n-N\right)$ is convex and decreasing on
$\left(0,1\right)$, we can conclude by Jensen's inequality that
\[
\mathbb{V}_{n}\geq\frac{1-p}{p}\left(n-N\right).
\]

\section{Proof of Lemma~\ref{lem:ignores-elements-with-delta}\label{lem:ignores-elements-with-delta-proof}}

First, we show that if a cardinality estimator verifies
$\left(\varepsilon,\delta\right)$-sketch privacy at a given cardinality, then
for each target $t$, we can find an explicit decomposition of the possible
outputs: the bound on information gain is satisfied for each possible output,
except on a density $\delta$. This is similar to the definition of
\emph{probabilistic} differential privacy in~\cite{machanavajjhala2008privacy}
and~\cite{gotz2012publishing}, except the decomposition depends on the choice of
$t$. We then use this decomposition to prove a variant of our negative result.

\begin{lem}\label{lem:epsilon-delta-decomposition}
If a cardinality estimator satisfies $\left(\varepsilon,\delta\right)$-sketch
privacy above cardinality $N$, then for every $n\geq N$ and $t\in\mathcal{U}$,
we can decompose the space of possible sketches
$\mathcal{M}=\mathcal{M}_{1}\uplus\mathcal{M}_{2}$ such that
$\mathbb{P}_{n}\left[M_{E}\in\mathcal{M}_{2} \mid t\in E\right]\leq2\delta$, and
for all $M\in\mathcal{M}_{1}$:
\[
\mathbb{P}_{n}\left[M_E=M \mid t\in E\right]
  \leq 2e^{\varepsilon}\cdot\mathbb{P}_{n}\left[M_E=M \mid t\notin E\right].
\]
\end{lem}

\begin{proof}
Suppose the cardinality estimator satisfies $\left(\varepsilon,\delta\right)$-sketch
privacy at cardinality $N$, and fix $n\geq N$ and $t\in\mathcal{U}$.
Let $\varepsilon^{\prime}=\varepsilon+\ln\left(2\right)$.

Let $\mathcal{M}_1$ be the set of outputs for which the privacy loss is higher
than $\varepsilon^{\prime}$. Formally, $\mathcal{M}_1$ is the set of sketches
$M$ that satisfy
\[
\mathbb{P}_{n}\left[M_E=M \mid t\in E\right]>e^{\varepsilon^{\prime}}\cdot\mathbb{P}_{n}\left[M_E=M \mid t\notin E\right].
\]
We show that $\mathcal{M}_1$ has a density bounded by $2\delta$.

Suppose for the sake of contradiction that this set has density at least $2\delta$ given
that $t\in E$:
\[
\mathbb{P}_{n}\left[M\in\mathcal{\mathcal{M}}_{1} \mid t\in E\right]\geq2\delta.
\]
We can sum the inequalities in $\mathcal{M}_1$ to obtain
\[
\mathbb{P}_{n}\left[M_{E}\in\mathcal{M}_{1} \mid t\in E\right]>e^{\varepsilon^{\prime}}\cdot\mathbb{P}_{n}\left[M_{E}\in\mathcal{M}_{1} \mid t\notin E\right].
\]
Averaging both inequalities, we get
\[
\mathbb{P}_{n}\left[M_{E}\in\mathcal{M}_{1} \mid t\in E\right]>e^{\varepsilon}\cdot\mathbb{P}_{n}\left[M_{E}\in\mathcal{M}_{1} \mid t\notin E\right]+\delta
\]
since $e^{\varepsilon^{\prime}}=e^{\varepsilon+\ln\left(2\right)}=2\cdot
e^{\varepsilon}$. This contradicts the hypothesis that the
cardinality estimator satisfied $\left(\varepsilon,\delta\right)$-sketch privacy
at cardinality $N$.

Thus, $\mathbb{P}_{n}\left[M\in\mathcal{M}_{1} \mid t\in E\right]<2\delta$ and
by the definition of $\mathcal{M}_{1}$, every output
$M\in\mathcal{M}_{2}=\mathcal{M}\mathbin{\backslash}\mathcal{M}_{1}$
verifies:
\[
\mathbb{P}_{n}\left[M_E=M \mid t\in E\right]
   \leq e^{\varepsilon^{\prime}}\cdot\mathbb{P}_{n}\left[M_E=M \mid t\notin E\right].
\]
\end{proof}

We can then use this decomposition to prove
Lemma~\ref{lem:ignores-elements-with-delta}.
Lemma~\ref{lem:epsilon-delta-decomposition} allows us to get two sets
$\mathcal{M}_{1}$ and $\mathcal{M}_{2}$ such that:
\begin{itemize}
\item $\mathbb{P}_{n}\left[M_{E}\in\mathcal{M}_{2} \mid t\in E\right]\leq2\delta$; and
\item for all $M\in\mathcal{M}_{1}$:
\[
\mathbb{P}_{n}\left[M_E=M \mid t\in E\right]\leq
    2e^{\varepsilon}\cdot\mathbb{P}_{n}\left[M_E=M \mid t\notin E\right].
\]
\end{itemize}

We decompose $\mathbb{P}_{n}\left[\mathit{add}\left(M_{E},t\right)=M_{E} \mid t\in E\right]$
into
\begin{multline*}
  \mathbb{P}_{n}\left[\mathit{add}\left(M_{E},t\right)=M_{E}\wedge M_{E}\in\mathcal{M}_{1} \mid t\in E\right]
  \\
  {} + \mathbb{P}_{n}\left[\mathit{add}\left(M_{E},t\right)=M_{E}\wedge M_{E}\in\mathcal{M}_{2} \mid t\in E\right].
\end{multline*}
The same reasoning as in the proof of Lemma~\ref{lem:ignores-elements}
gives
\begin{align*}
\lefteqn{\mathbb{P}_{n}\left[\mathit{add}\left(M_{E},t\right)=M_{E}\wedge M_{E}\in\mathcal{M}_{1} \mid t\in E\right]}\\
 & \leq 2e^{\varepsilon}\cdot\mathbb{P}_{n}\left[\mathit{add}\left(M_{E},t\right)=M_{E}\wedge M_{E}\in\mathcal{M}_{1} \mid t\notin E\right]\\
 & \leq 2e^{\varepsilon}\cdot\mathbb{P}_{n}\left[\mathit{add}\left(M_{E},t\right)=M_{E} \mid t\notin E\right]
\end{align*}
and since $\mathbb{P}_{n}\left[M_{E}\in\mathcal{M}_{2} \mid t\in E\right]\leq2\delta$,
we immediately have
\[
\mathbb{P}_{n}\left[\mathit{add}\left(M_{E},t\right)=M_{E}\wedge M\in\mathcal{M}_{2} \mid t\in E\right]\leq2\delta.
\]
We conclude that
\begin{multline*}
  \mathbb{P}_{n}\left[\mathit{add}\left(M_{E},t\right)=M_{E} \mid t\in E\right]
  \\
  {} \leq 
  2e^{\varepsilon}\cdot\mathbb{P}_{n}\left[\mathit{add}\left(M_{E},t\right)=M_{E} \mid t\notin E\right]+2\delta.
\end{multline*}
Now, Lemma~\ref{ce-property} gives $\mathbb{P}\left[\mathit{add}\left(M_{E},t\right)=M_{E} \mid t\in E\right]=1$,
and finally $\mathbb{P}_{n}\left[\mathit{add}\left(M_{E},t\right)=M_{E} \mid
t\notin E\right]\geq\left(\frac{1}{2}-\delta\right)\cdot e^{-\varepsilon}$.

\section{Proof of Lemma~\ref{lem:average-implies-delta}\label{lem:average-implies-delta-proof}}

Let $n\geq N$ and $\delta>0$. Suppose that a cardinality estimator
does not satisfy $\left(\frac{\varepsilon}{\delta},\delta\right)$-sketch
privacy at cardinality $n$. Then with probability strictly larger than $\delta$,
the output does not satisfy $\frac{\varepsilon}{\delta}$-sketch privacy.
Formally there is $\mathcal{M}_{2}\subseteq\mathcal{M}$ such that
$\mathbb{P}_{n}\left[M_{E}\in\mathcal{M}_{2} \mid t\in E\right]>\delta$, and
such that $\varepsilon_{n,t,M}>\frac{\varepsilon}{\delta}$ for all
$M\in\mathcal{M}_{2}$. Since all values of $\varepsilon_{n,t,M}$ are positive,
we have
\begin{align*}
\lefteqn{\sum_{M}\mathbb{P}_{n}\left[M_E=M \mid t\in E\right]\cdot\varepsilon_{n,t,M}}\\
 & \geq\sum_{M\in\mathcal{M}_{2}}\mathbb{P}_{n}\left[M_E=M \mid t\in E\right]\cdot\varepsilon_{n,t,M}\\
 & >\frac{\varepsilon}{\delta}\sum_{M\in\mathcal{M}_{2}}\mathbb{P}_{n}\left[M_E=M \mid t\in E\right]\\
 & >\varepsilon.
\end{align*}
Hence this cardinality estimator does not satisfy $\varepsilon$-sketch average
privacy at cardinality $n$.

\section{Proof of Theorem~\ref{thm:hll:average}\label{thm:hll:average-proof}}

First, let us formally define a HyperLogLog cardinality estimator.

\begin{defn}
Let $h$ be a uniformly distributed hash function. A \emph{HyperLogLog
cardinality estimator of parameter $p$} is defined as follows. A
sketch consists of a list of $2^{p}$ counters $C_{0}, \ldots, C_{2^{p}-1}$,
all initialized to $0$. When adding an element $e$ to the sketch,
we compute $h(e)$, and represent it as a binary string
$x = x_1 x_2 \ldots$.
Let $b(e)=\left\langle x_{1} \ldots x_{p}\right\rangle _{2}$, i.e., the integer represented by the binary digits $x_{1} \ldots x_{p}$,
and $\rho\left(e\right)$ be the position of the leftmost $1$-bit
in $x_{p+1}x_{p+2} \ldots$. Then we update counter $C_{b\left(e\right)}$
with $C_{b\left(e\right)}\leftarrow\max\left(C_{b\left(e\right)},\rho\left(e\right)\right)$.
\end{defn}

For example, suppose that $x=10001101\ldots$ and $p=2$, then
$b\left(e\right)=\left\langle 10\right\rangle _{2}=2$, and
$\rho\left(e\right)=3$ (the position of the leftmost $1$-bit in $001101\ldots$).
So we must look at the value for the counter $C_{2}$ and, if $C_{2}<3$, set
$C_{2}$ to $3$.

To simplify our analysis, we assume in this proof that
$\left|\mathcal{U}\right|$ is very large: for all reasonable values of $n$,
picking $n$ elements uniformly at random in $\mathcal{U}$ is the same as picking
a subset of $\mathcal{U}$ of size $n$ uniformly at random. In particular, we
approximate $\mathbb{P}_{n}\left[M_E=M \mid t\notin E\right]$ by
$\mathbb{P}_{n}\left[M_E=M\right]$.

First, we compute $\varepsilon_{n,t,M}$ for each $M$, by considering the counter
values of $M$. We then use this to determine $\varepsilon_{n,t}$, and averaging
them, deduce the desired result.

\begin{trivlist}
\item \emph{Step 1: Computing $\varepsilon_{n,t,M}$}

Let $t\in\mathcal{U}$, and $M\in\mathcal{M}$ such that $\mathbb{P}_{n}\left[M_E=M \mid t\in E\right]>0$.
Decompose the binary string $h\left(t\right)=x$ into two parts to get $b\left(t\right)$ and $\rho\left(t\right)$.
Denote by $C_{0}, \ldots, C_{2^{p}-1}$ the counters of $M$.
Note that  $M = M_E$ is characterized by two conditions:
\begin{itemize}
\item 
  $\text{REACH}_E(i)$ for all $0 \leq i < 2^p$, where $\text{REACH}_E(i)$ abbreviates
  $C_i > 0 \longrightarrow \exists e \mathbin{\in} E.\ b(e) = i \wedge \rho(e) = C_i$. 
  For each non-empty bucket, an element with this number of leading zeroes was added in this bucket.
  
\item
  $\forall e \in E.\ \rho(e) \leq C_{b(e)}$ (notation $\text{MAX}_{E}$).
  No element has more leading zeroes than the counter for its bucket.
\end{itemize}

Now, we compute the value of $\varepsilon_{n,t,M}$, depending on
the value of $C_{b\left(t\right)}$. Without loss of generality, we assume that $b(t)=0$.

\emph{Case 1}: Suppose $C_{0}=\rho(t)$.

We compute the probability of observing $M$ given $t\in E$.
$\text{REACH}_{E}\left(0\right)$ is already satisfied by $t$ as $C_{0}=\rho\left(t\right)$.
So for $M_{E}$ to be equal to $M$, $E$'s other elements only
have to satisfy $\text{REACH}\left(i\right)$ for $i\geq1$, and $\text{MAX}_{E}$
for all $i$:
\begin{equation*}
  \mathbb{P}_{n}[M_E=M \mid t\in E]
  = \mathbb{P}_{n-1}[\text{MAX}_{E} \wedge \forall i \mathbin{\ge}
  1.\,\text{REACH}_{E}(i)] \,.
\end{equation*}

Next, we compute the probability of observing $M$ given $t\notin E$.
This time, all elements of $E$ are chosen randomly, so
\begin{equation*}
  \begin{split}
    \mathbb{P}_{n}\left[M_E=M \mid t\notin E\right]
    &
    {}
    \simeq \mathbb{P}_{n}\left[M_E=M\right]
    \\
    &
    {} = \mathbb{P}_{n}\left[\text{MAX}_{E} \wedge \forall i.\,\text{REACH}_{E}(i)\right].
  \end{split}
\end{equation*}
We can decompose this condition: there is a witness $e\in E$ for $\text{REACH}_{E}(0)$ (i.e., $b(e)=0$ and $\rho(e)=C_{0}=\rho(t)$) and all others elements
satisfy the same condition as in the case $t \in E$, namely $\text{MAX}_{E} \wedge \forall i \mathbin{\ge} 1.\,\text{REACH}_{E}(i)$, as this is equivalent to $\text{MAX}_{E \mathbin{\backslash} \left\{e\right\}} \wedge \allowbreak \forall i \mathbin{\ge} 1.\,\text{REACH}_{E \mathbin{\backslash}\left\{e\right\}}(i)$ by the choice of $e$.
If $e$ is chosen uniformly in $\mathcal{U}$, then
\begin{itemize}
\item The probability of $b(e)=0$ is $2^{-p}$, since $h$ is a uniformly
distributed hash function.
\item The probability of $\rho(e)=\rho(t)$ is $2^{-\rho(t)}$:
since $h$ is a distributed hash function, if we note $h(t)=x_{1}\ldots x_{p}.x_{p+1}\ldots$,
then $x_{p+1}=1$ with probability $1/2$, $x_{p+1}x_{p+2}=01$ with
probability $1/4$, etc.
\end{itemize}
Thus, the probability that an element $e$ witnesses $\text{REACH}_{E}(0)$
is $2^{-p-\rho(t)}$ as $x_1 \ldots x_{p}$ are independent of $x_{p+1}x_{p+2}\ldots$.
Since the set $E$ has size $n$, there are $n$ possible chances that such an element is chosen:
the probability that at least one element witnesses $\text{REACH}_{E}(0)$
is $1-{\left(1-2^{-p-\rho\left(t\right)}\right)}^{n}$.
We can thus approximate $\mathbb{P}_{n}\left[M_E=M \mid t\notin E\right]$ by
\begin{equation*}
  \left(1-\!{\left(1-2^{-p-\rho\left(t\right)}\right)}^{\!n}\right)
  \cdot \mathbb{P}_{n-1}\left[\text{MAX}_{E} \wedge \forall i \mathbin{\ge} 1.\,\text{REACH}_{E}(i)\right]
\end{equation*}
and thus
\[
\frac{\mathbb{P}_{n}\left[M_E=M \mid t\in E\right]}{\mathbb{P}_{n}\left[M_E=M \mid t\notin E\right]}
\simeq
{\left(1-{\left(1-2^{-p-\rho\left(t\right)}\right)}^{n}\right)}^{-1} \,.
\]

\emph{Case 2}: Suppose $C_{0}>\rho\left(t\right)$.

We can compute the probabilities $\mathbb{P}_{n}\left[M_E=M \mid t\in E\right]$ and $\mathbb{P}_{n}\left[M_E=M \mid t\notin E\right]$
in a similar fashion.
\begin{align*}
  \mathbb{P}_{n}\left[M_E=M \mid t\in E\right]
  & {} =
  \mathbb{P}_{n-1}\left[\text{MAX}_{E} \wedge \forall i.\;\text{REACH}_{E}(i)\right]
  \\[1ex]
  \begin{split}
    \mathbb{P}_{n}\left[M_E=M \mid t\notin E\right]
    & 
    {} \simeq \mathbb{P}_{n}\left[M_E=M\right]
    \\
    &
    {} = \mathbb{P}_{n}\left[\text{MAX}_{E} \wedge \forall
      i.\;\text{REACH}_{E}(i)\right] \,.
  \end{split}
\end{align*}
Since (by hypothesis) $\mathbb{P}_{n}\left[M_E=M \mid t\in E\right]>0$,
there exist $n-1$ distinct elements of $\mathcal{U}$ which, together,
satisfy this condition $\text{MAX}_{E} \wedge \forall i.\;\text{REACH}_{E}(i)$. 
This allows us to bound $\mathbb{P}_{n}\left[M_E=M\right]$ from below.
Suppose one element $e$ of $E$ satisfies $b(e)=0$ and
$\rho(e)=\rho(t)$, and the $n-1$ other elements
in $E\mathbin{\backslash}\left\{ e\right\} $ satisfy the conditions $\text{MAX}_{E\mathbin{\backslash}\left\{ e\right\}}$
and $\text{REACH}_{E\mathbin{\backslash}\left\{ e\right\} }(i)$ for all $i$.
Then $E$ satisfies $\text{MAX}_{E}$ and $\text{REACH}_{E}(i)$ for all $i$.
The lower bound follows as before:
\newcommand{\gesim}{\underset{\mbox{$\sim$}}{>}}%
\newcommand{\lesim}{\underset{\mbox{$\sim$}}{<}}%
\begin{equation*}
  \begin{split}
    &\mathbb{P}_{n}\left[M_E=M \mid t\notin E\right]
    \gesim {}
    \\
    &
    \ \
    \left(1-{\left(1-2^{-p-\rho\left(t\right)}\right)}^{n}\right)
    \cdot \mathbb{P}_{n-1}\left[\text{MAX}_{E} \wedge \forall i.\,\text{REACH}_{E}(i)\right]
  \end{split}
\end{equation*}
and thus
\[
\frac{\mathbb{P}_{n}\left[M_E=M \mid t\in E\right]}{\mathbb{P}_{n}\left[M_E=M \mid t\notin E\right]}\lesim{\left(1-{\left(1-2^{-p-\rho\left(t\right)}\right)}^{n}\right)}^{-1}.
\]
We can use the results of both cases and immediately conclude that
\[
\varepsilon_{n,t,M}\lesim-\log\left(1-{\left(1-2^{-p-\rho\left(t\right)}\right)}^{n}\right)
\]
and that the equality holds if $M$ satisfies
$C_{b(t)}=\rho(t)$.

\item \emph{Step 2: Determining $\varepsilon_{n}$}

The previous reasoning shows that the worst case happens when the counter
corresponding to $t$'s bucket contains the number of leading zeroes of $t$,
i.e., when $\rho(t)=C_{b(t)}$:
\[
\varepsilon_{n,t}=\max_{M}\left(\varepsilon_{n,t,M}\right)\simeq-\log\left(1-{\left(1-2^{-p-\rho\left(t\right)}\right)}^{n}\right)
\,.
\]

We can then average this value over all $t$. Since the hash function
is uniformly distributed, $1/2$ of the values of $t$ satisfy $\rho\left(t\right)=1$,
$1/4$ satisfy $\rho\left(t\right)=1/4$, etc. Thus
\begin{equation*}
  \varepsilon_{n} \simeq -\sum_{k\geq1}2^{-k}\log\left(1-{\left(1-2^{-p-k}\right)}^{n}\right).
\end{equation*}
\end{trivlist}

\end{document}